\documentclass[10pt,reqno,a4paper]{amsart}
\usepackage[british]{babel}
\usepackage{amssymb,amstext,amsthm,eucal,bbm,mathrsfs,amscd}
\usepackage{tensor} 
\usepackage{physics} 
\usepackage{bbold}
\usepackage{stmaryrd}
\usepackage[hmargin=1in]{geometry}
\usepackage{array,color,slashed}
\usepackage[utf8]{inputenc}
\usepackage[small]{eulervm}
\usepackage{tgpagella}
\usepackage{enumitem,multicol,cite}
\usepackage{tikz,pgfplots}
\pgfplotsset{compat=1.13}
\usetikzlibrary{calc,arrows,cd}
\usepackage[unicode]{hyperref}
\hypersetup{%
  pdftitle   = {Killing superalgebras for lorentzian five-manifolds},
  pdfkeywords = {Spencer cohomology, Lie superalgebras, 5d, generalised Killing spinors},
  pdfauthor  = {Andrew Beckett, José Figueroa-O'Farrill},
  pdfcreator = {\LaTeX\ with package \flqq hyperref\frqq},
  colorlinks = true,
  linkcolor = blue,
  menucolor = blue,
  urlcolor = [rgb]{0.7,0,0},
  citecolor = [rgb]{0,0.7,0}
}
\theoremstyle{plain}
\newtheorem{lemma}{Lemma}
\newtheorem{proposition}[lemma]{Proposition}
\newtheorem{theorem}[lemma]{Theorem}
\newtheorem{corollary}[lemma]{Corollary}
\newtheorem{definition}[lemma]{Definition}
\theoremstyle{definition}

\let\d\partial
\newcommand{\half}{\tfrac{1}{2}}
\newcommand{\fg}{\mathfrak{g}}
\newcommand{\fgl}{\mathfrak{gl}}

\newcommand{\fh}{\mathfrak{h}}

\newcommand{\fp}{\mathfrak{p}}
\newcommand{\fs}{\mathfrak{s}}
\newcommand{\fso}{\mathfrak{so}}
\newcommand{\fstab}{\mathfrak{stab}}
\newcommand{\fsp}{\mathfrak{sp}}
\newcommand{\fsu}{\mathfrak{su}}

\newcommand{\fX}{\mathfrak{X}}
\newcommand{\fK}{\mathfrak{K}}
\newcommand{\fM}{\mathfrak{M}}
\newcommand{\RR}{\mathbb{R}}
\newcommand{\CC}{\mathbb{C}}
\newcommand{\HH}{\mathbb{H}}
\newcommand{\ZZ}{\mathbb{Z}}
\newcommand{\bbS}{\mathbb{S}}
\newcommand{\be}{\boldsymbol{e}}
\newcommand{\Cl}{C\ell}
\newcommand{\1}{\mathbb{1}} 
\newcommand{\id}{\operatorname{id}}

\newcommand{\sbar}{\overline{s}}

\newcommand{\B}{\mathsf{B}}
\ifdefined\C
\renewcommand{\C}{\mathsf{C}}
\else
\newcommand{\C}{\mathsf{C}}
\fi
\newcommand{\eD}{\mathscr{D}}

\newcommand{\eE}{\mathscr{E}}
\newcommand{\eH}{\mathscr{H}}
\newcommand{\eL}{\mathscr{L}}

\newcommand{\Spin}{\operatorname{Spin}}
\newcommand{\SO}{\operatorname{SO}}

\DeclareMathOperator{\Ric}{Ric}
\DeclareMathOperator{\End}{End}
\DeclareMathOperator{\Riem}{Riem}
\DeclareMathOperator{\Hom}{Hom}

\newcommand{\MUNCH}[1]{\relax}

\numberwithin{equation}{section}
\allowdisplaybreaks[1]
\begin{document}
\title[Five-dimensional Killing superalgebras]{Killing superalgebras for lorentzian five-manifolds}
\author[Beckett]{Andrew Beckett}
\author[Figueroa-O'Farrill]{José Figueroa-O'Farrill}
\address{Maxwell Institute and School of Mathematics, The University of
  Edinburgh, Edinburgh EH9 3FD, Scotland}
\date{\today}
\begin{abstract}
  We calculate the relevant Spencer cohomology of the minimal Poincaré
  superalgebra in 5 spacetime dimensions and use it to define Killing
  spinors via a connection on the spinor bundle of a 5-dimensional
  lorentzian spin manifold.  We give a definition of bosonic backgrounds
  in terms of this data. By imposing constraints on the curvature of the
  spinor connection, we recover the field equations of minimal
  (ungauged) 5-dimensional supergravity, but also find a set of field
  equations for an $\mathfrak{sp}(1)$-valued one-form which we interpret
  as the bosonic data of a class of rigid supersymmetric theories on
  curved backgrounds.  We define the Killing superalgebra of bosonic
  backgrounds and show that their existence is implied by the field
  equations. The maximally supersymmetric backgrounds are characterised
  and their Killing superalgebras are explicitly described as filtered
  deformations of the Poincaré superalgebra.
\end{abstract}
\thanks{EMPG-21-06}
\maketitle
\tableofcontents

\section{Introduction}
\label{sec:introduction}

The interplay between supersymmetry and geometry has a long and
beautiful history, but it is fair to say that we are still trying to
understand which geometries can support supersymmetric theories.
One reason is that the very notion of ``supersymmetric theory'' is
fluid.  If we take the conservative stance that a supersymmetric theory
is a field theory invariant under some supersymmetry algebra, then the
first question one needs to answer is which are the possible
supersymmetry algebras, to be followed by the analysis of their unitary
representations.  The former problem has not been completely solved,
whereas the latter problem is largely unexplored.

There are two main classes of supersymmetric theories, depending on
whether or not the supersymmetry is local.  The former are the
supergravity theories, many of which are related to low-energy limits of
superstring theories, whereas the latter are the rigidly supersymmetric
theories, which are the subject of much study today due to their rôle in
localisation in quantum field theory (see, e.g., \cite{Pestun:2016zxk}).
The two kinds of theories are closely related.  Indeed, one way to
construct rigidly supersymmetric theories, pioneered by Festuccia and
Seiberg \cite{Festuccia:2011ws} a decade ago, is to couple a
supergravity theory (with an off-shell formulation) to matter and then
to freeze the gravitational degrees of freedom (i.e., the fields in the
supergravity multiplet) in such a way that some supersymmetry is
preserved.  This results in a rigidly supersymmetric theory for the
matter multiplet.  Supergravity theories with an off-shell formulation
are rare, however, and hence it is desirable to find alternative means
to constructing rigidly supersymmetric field theories.

Whereas in the Festuccia--Seiberg approach it is neither essential nor
indeed desirable for the fields in the supergravity multiplet to be
on-shell, but only for them to preserve some supersymmetry, in the strict
context of supergravity the interesting geometries are the
supersymmetric bosonic backgrounds.  A bosonic background is a solution
of the supergravity field equations where the fermionic fields have been
put to zero.  Bosonic backgrounds have very rich geometries, being after
all examples of (higher-dimensional) general relativity coupled to
matter.

A particularly interesting and rich subclass of bosonic backgrounds are
those which preserve some supersymmetry. Since fermions are set to zero,
the variation of any bosonic field under supersymmetry is automatically
zero, but not so for the variation of the fermionic fields.  In
particular, the characteristic property of a supergravity theory is that
the variation of the gravitino $\Psi$ under a supersymmetry
transformation with spinor field parameter $\varepsilon$ takes the form
$\delta_\varepsilon \Psi = \eD\varepsilon$, where $\eD$ is a connection
on spinors which, on a bosonic background, includes terms depending on
the additional bosonic fields in the supergravity multiplet.  For such a
transformation to preserve a bosonic background, this variation must
vanish on that background; in other words, $\varepsilon$ must be
parallel with respect to $\eD$. The condition $\eD\varepsilon = 0$,
possibly augmented by algebraic conditions coming from the
supersymmetric variations of any other fermionic fields in the
supergravity multiplet, is the Killing spinor equation. The spinors
$\varepsilon$ obeying it are called Killing spinors, because
squaring such a spinor gives rise to a Killing vector field, known
as its Dirac current.

It is always the case, perhaps after imposing some additional conditions
on the bosonic fields, that the Dirac current of a Killing spinor
preserves the other bosonic fields of the background, and hence it
preserves the connection $\eD$.  This implies that such Killing vectors
preserve the space of Killing spinors, together with which they generate a Lie
superalgebra known as the Killing superalgebra of the background
\cite{FigueroaO'Farrill:2004mx}. The Killing superalgebra is a useful
algebraic invariant of a supersymmetric supergravity background, and one
consequence of the homogeneity theorem \cite{FigueroaO'Farrill:2012fp}
is that it determines a ($>\tfrac12$)-BPS\footnote{i.e., any background
  where the dimension of the space of Killing spinors is more than half
  the rank of the spinor bundle.} background up to local isometry.  This
was proved in \cite{Figueroa-OFarrill:2016khp} for eleven-dimensional
supergravity, but it holds in general for any background for which the
Killing superalgebra is transitive.

The construction of the Killing superalgebra suggests that all we need
in order to identify which geometries can support rigid supersymmetry is
a suitable notion of Killing spinor: one which guarantees that the
Killing spinors generate a Lie superalgebra.\footnote{For a different
  approach to this problem, based on the algebraic classification of
  supersymmetric extensions of known spacetime algebras, and the
  superisation of the corresponding homogeneous spacetimes, see, e.g.,
  \cite{Figueroa-OFarrill:2019ucc}.}  For example, in the standard
Poincaré supersymmetry on Minkowski spacetime, Killing spinors are
parallel with respect to the spin connection, whereas in AdS
supersymmetry \cite{Zumino:1977av} Killing spinors are so-called
geometric Killing spinors, satisfying $\nabla_X \varepsilon = \lambda X
\cdot \varepsilon$, for some constant $\lambda$ related to the curvature
of AdS.  Geometric Killing spinors were also used in the pioneering work
of Blau's \cite{Blau:2000xg} for the construction of rigidly
supersymmetric gauge theories.  Parallel and geometric Killing spinors
are intrinsic notions on any spin manifold, but the resulting theories
are not too different from Poincaré supersymmetry.  To make further
progress we need to consider other notions of Killing spinors.

If we assume that the definition of a Killing spinor is that it be
parallel with respect to a suitable connection in the spinor bundle
(possibly augmented by algebraic -- i.e., non-differential --
constraints), then a straightforward generalisation of the result in
\cite{Figueroa-OFarrill:2016khp} for the Killing superalgebra of
eleven-dimensional supergravity backgrounds shows that the resulting
superalgebra has a special algebraic structure.  Namely, it is naturally
filtered in such a way that the associated graded superalgebra is a
graded subalgebra of the Poincaré superalgebra.  We say that it is a
filtered subdeformation of the Poincaré superalgebra.

Such deformations are governed by the (generalised) Spencer cohomology of
graded superalgebras \cite{MR1688484,MR2141501}, and so calculation of
the relevant Spencer cohomology groups for the Poincaré superalgebra is
the first step in this analysis. The Spencer cohomology not only determines
the filtered deformations of the Poincaré superalgebra, it also gives
the expression for the connection defining the notion of a Killing
spinor.  In some cases, such as the $D=11$
\cite{Figueroa-OFarrill:2015rfh,Figueroa-O'Farrill:2015utu} and minimal $D=4$
\cite{deMedeiros:2016srz} Poincaré superalgebras, one obtains
precisely the connection $\eD$ of a supergravity theory, but in other
cases, such as the minimal $D=6$ \cite{deMedeiros:2018ooy} Poincaré
superalgebra, the Spencer cohomology is richer: additional bosonic fields
may be turned on, and the definitions of Killing spinors, supersymmetric
backgrounds and Killing superalgebras may be consistently generalised to
accommodate them. The existence of such generalisations is intriguing,
not least because they provide curved backgrounds for rigidly
supersymmetric theories which do not appear to be attainable via
supergravity.

In the present paper we discuss the Spencer cohomology of the minimal
$D=5$ Poincaré superalgebra.  Our motivation is two-fold.  On the one
hand, it is an intermediate case between two similar calculations:
minimal $D=4$ and $D=6$ Poincaré superalgebras, and provides a useful
additional datapoint in framing a conjecture about the behaviour of
Spencer cohomology under dimensional reduction.  A second motivation is
that we may then go on to study supersymmetric reductions of the
geometries admitting maximal supersymmetry to four dimensions and
perhaps in this way obtain novel four-dimensional lorentzian and
riemannian spin manifolds admitting rigid supersymmetry.  This is the
subject of ongoing work.

We find that the relevant Spencer cohomology group
\(H^{2,2}(\fs_-,\fs)\) is parametrised by a two-form, which is expected
from supergravity, and an \(\fsp(1)\)-valued one-form, which does not
correspond to any supergravity field. This is reminiscent of the Spencer
data from minimal \(D=6\), which also includes an additional
\(\fsp(1)\)-valued one-form \cite{deMedeiros:2018ooy}. As in the
6-dimensional case, after using Spencer cocycles to define a connection
on spinors, by imposing a constraint (the vanishing of the Clifford trace)
on the curvature of that connection, the bosonic equations of motion for
supergravity can be recovered along with an additional set of field
equations for the one-form.

This paper is organised as follows. In Section~\ref{sec:poinc-super}, we
introduce our conventions and some identities before describing the
minimal Poincaré superalgebra \(\fs\) in 5 dimensions as a graded Lie
superalgebra.  Our calculation of the Spencer cohomology
\(H^{2,2}(\fs_-,\fs)\) is given in Section~\ref{sec:spencer-cohomology}
and culminates in Theorem~\ref{thm:Spencer-calculation}. This data is
interpreted geometrically in Section~\ref{sec:zero-curv-eq} by using it
to define a connection \(\eD\) on spinors, as well as an associated
notion of Killing spinors, in a suitable geometric setting (a bosonic
background). The curvature of \(\eD\) is explicitly calculated in terms
of the bosonic background fields and various conditions on it are
characterised.  Theorem~\ref{thm:gamma-trace-zero} characterises those
geometries where the Clifford trace of the curvature vanishes, whereas
Theorem~\ref{thm:maximalsusy} characterises those geometries with
vanishing curvature.  Section~\ref{sec:killing-superalg} is concerned
with Killing superalgebras: the spinorial Lie derivative is defined and
some of its properties described; some properties of Killing spinors are
derived; then finally Killing superalgebras are defined, and their
existence is proven (Theorem~\ref{thm:killing-superalg-exist}).  In
Section~\ref{sec:max-susy}, we describe the maximally supersymmetric
backgrounds explicitly.  As can be read in
Theorem~\ref{thm:max-susy-bgs}, these fall into two branches: the first
coincides with maximally supersymmetric supergravity backgrounds, making
contact with known results \cite{Gauntlett:2002nw,Chamseddine:2003yy};
the second is characterised by the existence of an \(\fsp(1)\)-valued
one-form. We determine backgrounds belonging to the second branch,
noting the resemblance to the 6-dimensional case
\cite{deMedeiros:2018ooy}, and we describe the Killing superalgebras in
both branches explicitly as filtered deformations of \(\fs\).
Appendix~\ref{sec:tensor-id} is a compilation of combinatorial tensor
identities used in geometric calculations.

\section{The Poincaré superalgebra}
\label{sec:poinc-super}

In this section we set up our conventions and introduce the Poincaré superalgebra.

\subsection{Spinorial conventions}
\label{sec:spin-conv}

Let $(V,\eta)$ be a five-dimensional (mostly minus) lorentzian vector
space.  We will let $\flat : V \to V^*$ denote the musical isomorphism
sending $v$ to $v^\flat$, where
\begin{equation}
  v^\flat(w) = \eta(v,w).
\end{equation}
We define $\fso(V)$ to be the Lie algebra of $\eta$-skew-symmetric
endomorphisms of $V$:
\begin{equation}
  \fso(V) = \left\{ A : V \to V ~ \middle | ~ \eta(A v, w) = - \eta(v, Aw)
    \quad\forall v,w \in V\right\}.
\end{equation}
There is a vector space (in fact, an $\fso(V)$-module) isomorphism
$\fso(V) \cong \wedge^2 V$.  If $A \in \fso(V)$, we define $\omega_A
\in \wedge^2V$ by
\begin{equation}
  A v = - \iota_{v^\flat}\omega_A.
\end{equation}
Conversely, if $\omega \in \wedge^2V$, we define $A_\omega \in \fso(V)$
by the same relationship: namely,
\begin{equation}
  A_\omega v = - \iota_{v^\flat} \omega.
\end{equation}
It then follows that these two maps are mutual inverses: $A_{\omega_A} =
A$ and $\omega_{A_\omega} = \omega$.  Relative to an orthonormal basis
$\be_\mu$ for $V$, with $\eta(\be_\mu,\be_\nu) = \eta_{\mu\nu}$, we find
that
\begin{equation}
  \omega_A = \tfrac{1}{2} A^{\mu\nu} \be_\mu \wedge
  \be_\nu\qquad\text{where}\qquad A \be_\mu = \be_\nu A\indices{^\nu_\mu},
\end{equation}
and indices are lowered and raised with $\eta_{\mu\nu}$ and its inverse
$\eta^{\mu\nu}$.

We define the Clifford algebra $\Cl(V)$ by the Clifford relations (notice the sign!)
\begin{equation}
  v \cdot v = \eta(v,v) \1.
\end{equation}
As a real associative algebra,
$\Cl(V) \cong \End(\Sigma) \oplus \End(\Sigma')$, where $\Sigma$ and
$\Sigma'$ are two inequivalent irreducible Clifford modules, which are
two-dimensional quaternionic (right) vector spaces. They are
distinguished by the action of the centre of $\Cl(V)$. The centre is
spanned by the identity and the volume element, defined by the
Levi-Civita symbol $\epsilon_{\mu\nu\rho\sigma\tau}$ normalised to
$\epsilon_{01234} = 1$.

On the Clifford module $\Sigma$ the volume element acts like the
identity endomorphism $\id_\Sigma$, whereas on $\Sigma'$ it acts like
$-\id_{\Sigma'}$. In other words, the centre of $\Cl(V)$ acts trivially on
$\Sigma$ and non-trivially on $\Sigma'$. We will work with $\Sigma$ from
now on.  We will also use the notation $\1$ for the identity
endomorphism of $\Sigma$.

Under the representation homomorphism $\Cl(V) \to \End(\Sigma)$ the basis
element $\be_\mu$ is represented by the endomorphism $\Gamma_\mu$.
These endomorphisms satisfy the Clifford relation
\begin{equation}
  \Gamma_\mu \Gamma_\nu + \Gamma_\nu \Gamma_\mu = 2 \eta_{\mu\nu} \1.
\end{equation}
In addition, they obey
\begin{equation}
  \Gamma_\mu \Gamma_\nu = \Gamma_{\mu\nu} + \eta_{\mu\nu} \1,
\end{equation}
with $\Gamma_{\mu\nu} = \tfrac12[\Gamma_\mu,\Gamma_\nu]$, et cetera.
Since the volume element acts trivially and Hodge duality is implemented
by multiplication by the volume element, a basis for $\End(\Sigma)$ is
given by $(\1, \Gamma_\mu,\Gamma_{\mu\nu})$.  Indeed, we have the
following useful identities in $\End(\Sigma)$ for the other
skew-symmetric products of the $\Gamma_\mu$:
\begin{equation}
  \Gamma_{\mu\nu\rho} = -\tfrac12 \epsilon_{\mu\nu\rho\sigma\tau}
  \Gamma^{\sigma\tau}, \qquad
  \Gamma_{\mu\nu\rho\tau} = \epsilon_{\mu\nu\rho\sigma\tau} \Gamma^\tau
  \qquad\text{and}\qquad
  \Gamma_{\mu\nu\rho\sigma\tau} = \epsilon_{\mu\nu\rho\sigma\tau} \1.
\end{equation}

\begin{lemma}\label{lemma:gamma-traces}
  The following identities between gamma matrices hold:
  \begin{equation}
    \begin{aligned}
      \Gamma_\mu\Gamma^\mu &= 5 \1\\
      \Gamma_\mu \Gamma_\nu \Gamma^\mu &= -3 \Gamma_\nu\\
      \Gamma_\mu \Gamma_{\nu\rho} \Gamma^\mu &= \Gamma_{\nu\rho}
    \end{aligned}
    \qquad\qquad
    \begin{aligned}
      \tfrac12 \Gamma_{\mu\nu} \Gamma^{\mu\nu} &= -10 \1\\
      \tfrac12 \Gamma_{\mu\nu} \Gamma_\rho \Gamma^{\mu\nu} &= -2
      \Gamma_\rho\\
      \tfrac12 \Gamma_{\mu\nu} \Gamma_{\rho\sigma} \Gamma^{\mu\nu} &= 2 \Gamma_{\rho\sigma}.
    \end{aligned}
  \end{equation}
\end{lemma}

If $A \in \fso(V)$, its action on $\Sigma$ is given by Clifford product with
$\tfrac12 \omega_A$:
\begin{equation} \label{eq:so-spinor-action}
  A s = \tfrac12 \omega_A \cdot s = \tfrac14 A^{\mu\nu} \Gamma_{\mu\nu}
  s.
\end{equation}

Let $\Delta$ denote a one-dimensional quaternionic (right) vector space,
which we think as a two-dimensional complex vector space. We can
similarly think of $\Sigma$ as a four-dimensional complex vector space
with a quaternionic structure, so that the (complex) tensor product
$\Sigma \otimes_\CC \Delta$ has a real structure and hence it is the
complexification of a real eight-dimensional representation we denote by
$S$. As a real vector space, \(S\) is just $\Sigma$ when we restrict scalars
from $\HH$ to $\RR$. Spinors in $S$ are ``symplectic Majorana'' spinors
and we choose to view them as pairs $s^A\in\Sigma$, where $A = 1,2$, subject to a
symplectic Majorana condition which uses the symplectic structures on
$\Sigma$ and $\Delta$, respectively, and which we will write presently.
Let us normalise $\epsilon_{AB}$ so that $\epsilon_{12} = + 1$ and
$\epsilon^{12} = +1$.  We raise and lower indices with $\epsilon$
according to the conventions: $\epsilon_{AB} X^B{}_C= X_{AC}$ and
$\epsilon^{AB} X_{AC} = X^B{}_C$, et cetera. The symplectic Majorana
condition on the pair $s^A$ is such that
\begin{equation}
  \label{eq:symplectic-majorana}
  (s^A)^* = \epsilon_{AB} \B s^B,
\end{equation}
where $\B$ obeys
\begin{equation}
  \B \Gamma_\mu = \Gamma_\mu^* \B
\end{equation}
and normalised to $\B^\dagger \B = \1$.

It follows that the gamma matrices preserve \(S\), hence so does the
16-dimensional real associative subalgebra
\(\RR\Gamma\subset \End_\CC(\Sigma)\) they generate. Note that this
algebra is simply the image of \(\Cl(V)\) under the representation on
\(\Sigma\). We denote by \(\End_\HH\Delta\) the real associative
subalgebra of \(\End_\CC\Delta\) which preserves the quaternionic
structure on \(\Delta\). As a real vector space,
\(\End_\HH\Delta = \RR\1\oplus\fsp(1)\). If we use the symplectic form
to lower indices and identify \(\Delta^*\cong\Delta\), we also identify
\(\End_\CC\Delta\cong \otimes^2 \Delta =
\odot^2\Delta\oplus\wedge^2\Delta\). Under this identification,
\(\fsp(1)\) is the quaternionic structure-preserving part of
\(\odot^2\Delta\) and \(\RR\1\cong \RR\epsilon\) is that of
\(\wedge^2\Delta\). We can identify \(\End_\RR(S)\) as the real
subalgebra of \(\End(\Sigma\otimes_\CC \Delta)\) which preserves \(S\),
and it follows from the discussion above that
\(\End_\RR(S)=\RR\Gamma\otimes(\RR\1\oplus\fsp(1))\). The inclusion in
one direction is clear; the other then follows by a dimension count.

The symplectic structure $\C$ on $\Sigma$ is real and satisfies
\begin{equation}
  \label{eq:CC}
  \C\Gamma_\mu = \Gamma_\mu^T \C \implies \C\Gamma_{\mu\nu} = -
  \Gamma_{\mu\nu}^T \C.
\end{equation}
We define for $s_1^A, s_2^A \in \Sigma$,
\begin{equation}
  \label{eq:IP}
  \sbar_1^A s_2^B :=
  \left(s_1^A\right)^T \C s_2^B.
\end{equation}

\begin{lemma}\label{sec:idents}
  The following identities hold for all $s_1^A,
  s_2^A \in \Sigma$:
  \begin{enumerate}
  \item $\sbar_1^A s_2^B  = -
    \sbar_2^B s_1^A$,
  \item $\sbar_1^A \Gamma_\mu s_2^B  = -
    \sbar_2^B \Gamma_\mu s_1^A$,
  \item $\sbar_1^A \Gamma_{\mu\nu} s_2^B  = 
    \sbar_2^B \Gamma_{\mu\nu} s_1^A$.
  \end{enumerate}
\end{lemma}

This allows us to define a symmetric inner product on $S$:
\begin{equation}
  \label{eq:IPS}
  \left<s_1, s_2\right> :=\epsilon_{AB} \sbar_1^A s_2^B.
\end{equation}
It follows from Lemma~\ref{sec:idents} that for all $v \in V$,
\begin{equation}
  \left<s_1, v \cdot s_2\right> = \left<v \cdot s_1, s_2\right>,
\end{equation}
and from this it follows that it is $\fso(V)$-invariant:
\begin{equation}
  \left<A s_1, s_2\right> = - \left<s_1, A s_2\right>,
\end{equation}
for all $A \in \fso(V)$.

There are a number of bilinears we can make from spinors.  If $s \in S$,
then we define
\begin{equation}\label{eq:spinor-bilinears}
  \mu_s := \epsilon_{AB} \sbar^A s^B, \quad \kappa_s^\mu := \epsilon_{AB}
  \sbar^A \Gamma^\mu s^B \qquad\text{and}\qquad \omega^{AB}_{s,\mu\nu} =
  \sbar^A \Gamma_{\mu\nu} s^B.
\end{equation}
The map  $\kappa: S \to V$ sending $s \mapsto \kappa_s$ defines the
\textbf{Dirac current} of $s$ to be the unique vector $\kappa_s \in V$
such that for all $v \in V$,
\begin{equation}
  \label{eq:Dirac}
  \eta(\kappa_s,v) = \left<s, v \cdot s\right>.
\end{equation}
Similarly, $\mu_s = \left<s,s\right>$ more invariantly and also for all
$v,w \in V$, 
\begin{equation}
  \omega^{AB}_s(v,w)  = \tfrac12 \sbar^A [v,w]\cdot s^B,
\end{equation}
where $[v,w]$ is the Clifford commutator.  The following is the result
of a calculation in an explicit realisation.

\begin{lemma}[Reality conditions]\label{lemma:bil-real-cond}
  For $s$ symplectic Majorana, it follows that $\mu_s$ and $\kappa_s$
  are real, whereas $\left(\omega_s^{11}\right)^* = \omega_s^{22}$ and
  $\left(\omega_s^{12}\right)^* = - \omega_s^{12}$.
\end{lemma}

The proof of the following Fierz identities is routine.  We simply remark
that $\tr \1 = 4$ since we are working in a four-dimensional complex
vector space (with a quaternionic structure).

\begin{lemma}[Fierz identities]
  For all $s^A_1,s^B_2 \in \Sigma$,
  \begin{equation}
    s_1^A \sbar_2^B = -\tfrac14 \sbar_1^A s_2^B \1 - \tfrac14  \sbar_1^A
    \Gamma^\mu s_2^B \Gamma_\mu - \tfrac18  \sbar_1^A \Gamma^{\mu\nu}
    s_2^B \Gamma_{\mu\nu}.
  \end{equation}
  In particular, taking $s_1 = s_2 = s$ we arrive at
  \begin{equation}
    \label{eq:fierztoo}
    s^A \sbar^B =-\tfrac18 \epsilon^{AB} \mu_s \1 - \tfrac18 \epsilon^{AB}
    \kappa_s - \tfrac14 \omega_s^{AB},
  \end{equation}
  where we define $\omega_s^{AB} = \tfrac12 \omega_{s,\mu\nu}^{AB}
  \Gamma^{\mu\nu}$.
\end{lemma}

From this latter Fierz identity there follow some useful relations between the
bilinears $\mu_s$, $\kappa_s$ and $\omega^{AB}_s$ defined above.

\begin{lemma}\label{sec:bilidents}
  Let $s \in S$ and let $\mu$, $\kappa$ and $\omega^{AB}$ be the
  corresponding bilinears.  Then they satisfy the following relations:
  \begin{enumerate}[label=(\roman*)]
  \item $\kappa \cdot s^A = \mu s^A$
  \item $3 \mu s^C + \epsilon_{AB} \omega^{CA} \cdot s^B = 0$
  \item $\omega^{(AB} \cdot s^{C)} = 0$
  \item $\omega^{AB} (\kappa, -) = 0$
  \end{enumerate}
\end{lemma}

\begin{proof}
  As a result of the Fierz identity \eqref{eq:fierztoo} we find that
  \begin{equation}
    \begin{split}
      \mu s^C &= \epsilon_{AB} \sbar^A s^B s^C\\
      &= \epsilon_{AB} s^C\sbar^A s^B\\
      &= \epsilon_{AB} \left(-\tfrac18 \epsilon^{CA} \mu \1 - \tfrac18 \epsilon^{CA}
        \kappa - \tfrac14 \omega^{CA}\right) \cdot s^B\\
      &= \tfrac18 \mu s^C + \tfrac18 \kappa \cdot s^C - \tfrac14
      \epsilon_{AB} \omega^{CA} \cdot s^B
    \end{split}
  \end{equation}
  or, equivalently,
  \begin{equation}\label{eq:nuS}
    \tfrac78 \mu s^C - \tfrac18 \kappa \cdot s^C + \tfrac14
    \epsilon_{AB} \omega^{CA} \cdot s^B = 0.
  \end{equation}
  Similarly, by calculating $\kappa \cdot s^C$ and using the Fierz
  identity~\eqref{eq:fierztoo}, we arrive at
  \begin{equation}\label{eq:kappaS}
    \tfrac58 \mu s^C - \tfrac{11}8 \kappa \cdot s^C - \tfrac14
    \epsilon_{AB} \omega^{CA} \cdot s^B = 0.
  \end{equation}
  Adding equations~\eqref{eq:nuS} and \eqref{eq:kappaS}, we obtain (i)
  and, plugging (i) back into either of those two equations, we obtain
  (ii).  To obtain (iii), we calculate $\omega^{AB} \cdot s^C$ again 
  using the Fierz identity \eqref{eq:fierztoo}, obtaining
  \begin{equation}
    \label{eq:omegaS}
    \omega^{AB} \cdot s^C + \tfrac12 \omega^{CA}\cdot s^B + \tfrac54
    \epsilon^{AC} \mu s^B + \tfrac14 \epsilon^{AC} \kappa \cdot s^B = 0.
  \end{equation}
  Using (i), we rewrite this as
  \begin{equation}
    \omega^{AB} \cdot s^C + \tfrac12 \omega^{CA}\cdot s^B + \tfrac32
    \epsilon^{AC} \mu s^B = 0.
  \end{equation}
  This equation decomposes into two equations by (skew-)symmetrising in
  $AC$.  The skew-symmetric component
  \begin{equation}
    \omega^{B[A} \cdot s^{C]} + \tfrac32 \epsilon^{AC} \mu s^B = 0
  \end{equation}
  vanishes identically by (ii), whereas the symmetric component
  \begin{equation}
    \tfrac12 \omega^{BA} \cdot s^C + \tfrac12 \omega^{BC} \cdot
    s^A  + \tfrac12 \omega^{CA}\cdot s^B = 0
  \end{equation}
  is precisely (iii).  Finally, to obtain (iv), we notice that for any $w
  \in V$,
  \begin{equation}
    \begin{split}
      \omega^{AB}(\kappa, w) &= \tfrac12 \sbar^A [\kappa, w] \cdot s^B\\
      &= \tfrac12 \sbar^A \kappa \cdot w \cdot s^B - \tfrac12 \sbar^A w
      \cdot \kappa \cdot s^B \\
      &= \tfrac12 \overline{\kappa \cdot s}^A w \cdot s^B - \tfrac12 \sbar^A w
      \cdot \kappa \cdot s^B,
    \end{split}
  \end{equation}
  which vanishes by (i) and where we have used that:
  \begin{equation}
    \sbar^A \kappa = (s^A)^T C \kappa = (s^A)^T \kappa^T C = (\kappa
    \cdot s^A)^T C = \overline{\kappa\cdot s}^A.
  \end{equation}
\end{proof}

It follows from (ii) and (iii) in Lemma~\ref{sec:bilidents} that
$\omega^{AB} \cdot s^C = - \mu \left(s^A \epsilon^{BC} + s^B
  \epsilon^{AC} \right) = - 2\mu s^{(A} \epsilon^{B)C}$. In addition, it
follows from (i) that $\eta(\kappa,\kappa) = \mu^2 \geq 0$, so that
$\kappa$ is causal: timelike if $\left<s,s\right> \neq 0$ and null if
$\left<s,s\right> =0$, and it follows from (iv) that the elements of
$\fso(V)$ corresponding to $\omega^{AB}$ leave $\kappa$ invariant. We
also have \(\omega_{\mu\nu}^{AB}\omega^{\mu\nu}_{AB}=6\mu^2\geq 0\) by
(i), (ii), and a further application of the Fierz identity.  There are
other identities which can be derived, such as
\begin{equation}
  \omega^{AB}_{[\mu\nu} \kappa_{\rho]} = -\tfrac16 \mu
  \epsilon_{\mu\nu\rho\sigma\tau} \omega^{AB\,\sigma\tau}.
\end{equation}

The $2$-forms $\omega^{AB}$ associated with $s$ define spinor
endomorphisms
\begin{equation}
  \hat\omega^{AB} := \tfrac14 \omega^{AB}_{\mu\nu} \Gamma^{\mu\nu},
\end{equation}
whose commutators satisfy:
\begin{equation}
  [\hat\omega^{12}, \hat\omega^{11}] = \mu \hat\omega^{11},\qquad
  [\hat\omega^{12}, \hat\omega^{22}] = -\mu \hat\omega^{22} \qquad\text{and}\qquad
  [\hat\omega^{11}, \hat\omega^{22}] = -2\mu \hat\omega^{12},
\end{equation}
defining, for $\mu \neq 0$, the complex Lie algebra
$\mathfrak{sl}(2,\CC)$.  However, the symplectic Majorana condition on
$s$ picks out a real form corresponding to the real $2$-forms; using 
Lemma~\ref{lemma:bil-real-cond}, this is the real span of the following 
spinor endomorphisms:
\begin{equation}
    L_1 := \tfrac12 (\hat\omega^{11} + \hat\omega^{22}),\qquad
    L_2 := \tfrac{i}2 (\hat\omega^{11} - \hat\omega^{22}) \qquad\text{and}\qquad
    L_3 := i \hat\omega^{12}.
\end{equation}
These satisfy
\begin{equation}
  [L_i,L_j] = \mu \epsilon_{ijk} L_k,
\end{equation}
for $i,j,k = 1,2,3$.  In summary we have proved the following.

\begin{proposition}
  Let \(s\in S\) and \(\mu\), \(\kappa\) and
  \(\omega^{AB}\) be the corresponding bilinears. Then the 2-forms
  \(\omega^{AB}\) define a Lie subalgebra of the stabiliser of
  \(\kappa\) in \(\fso(V)\) under the isomorphism
  \(\fso(V)\cong \wedge^2 V\).  If \(\kappa\) is null (equivalently,
  $\mu = 0$), this subalgebra is abelian and if \(\kappa\) is timelike,
  the subalgebra is isomorphic to $\mathfrak{sp}(1)$.
\end{proposition}

\subsection{The Poincaré superalgebra}
\label{sec:poin-super}

The Poincaré algebra $\fp(V)$ is the Lie algebra of isometries of
$(V,\eta)$.  It is a $15$-dimensional $\ZZ$-graded Lie algebra with
underlying vector space
\begin{equation}
  \fp(V) = \fp_{-2} \oplus \fp_0 = V \oplus \fso(V)
\end{equation}
and Lie brackets
\begin{equation}
  \label{eq:PLA}
  [A,B] = AB - BA \qquad [A,v] = Av \qquad\text{and}\qquad [v,w] = 0
\end{equation}
for all $v,w \in V$, and $A,B \in \fso(V)$.

The $d=5$ Poincaré superalgebra $\fs(V)$ is the minimal
superalgebra extending the Poincaré algebra of $V$.  In fact, the even Lie
algebra $\fs_{\bar 0} = \fp(V)$ is the Poincaré algebra.  The odd
subspace $\fs_{\bar 1} = S$, the real eight-dimensional Clifford module
$S$.  The underlying vector space of the $\ZZ$-graded Lie superalgebra
$\fs(V)$ is
\begin{equation}
  \fs(V) = \fs_{-2} \oplus \fs_{-1} \oplus \fs_0 = V \oplus S \oplus \fso(V)
\end{equation}
and Lie brackets given by those in \eqref{eq:PLA} and
\begin{equation}
  \label{eq:PSA}
  [A,s] = \half \omega_A \cdot s \qquad [v,s] = 0 \qquad\text{and}\qquad
  [s,s] = \kappa_s
\end{equation}
for all $s \in S$, $v \in V$ and $A \in \fso(V)$. In the remainder of the paper, 
we will denote the Poincaré superalgebra simply by \(\fs\), leaving \(V\) 
implicit in the notation.

\section{Spencer cohomology}
\label{sec:spencer-cohomology}

In this section we describe our calculation of the Spencer cohomology of
the $d=5$ Poincaré superalgebra.

\subsection{Lie algebra cohomology}
\label{sec:chevalley-eilenberg}

We start by briefly recalling the Chevalley--Eilenberg complex computing
the cohomology of a (finite-dimensional) Lie algebra $\fg$ relative to a
module $\fM$.  Let $G$ be the simply-connected Lie group with
Lie algebra $\fg$.  The exterior derivative of a left-invariant
differential form on $G$ is itself left-invariant and hence the
left-invariant forms on $G$ define a sub-complex of the de~Rham complex
of $G$.  Furthermore, since a left-invariant form is uniquely determined
by its value at the identity, we may localise the sub-complex of
left-invariant forms at the identity, resulting in a differential
complex $(C^\bullet(\fg)= \wedge^\bullet\fg^*, \d)$ known as the
Chevalley--Eilenberg complex.  Its cohomology $H^\bullet(\fg)$ is the
Lie algebra cohomology of $\fg$ or, more precisely, the Lie algebra
cohomology of $\fg$ with values in the trivial module $\RR$.

Let $\rho : \fg \to \fgl(\fM)$ be a representation of $\fg$ and let us
also denote by $\rho$ the corresponding representation of the
simply-connected $G$.  Letting $L_g: G \to G$ denote the left
multiplication by $g \in G$, we may consider the $G$-equivariant
$\fM$-valued differential forms on $G$: $\fM$-valued differential forms
$\omega$ which obey $L_g^* \omega = \rho(g) \circ \omega$ for all
$g \in G$.  Every such $\omega$ localises at the identity to a linear
map $\wedge^\bullet\fg \to \fM$.  The space
$C^\bullet(\fg,\fM):= \fM \otimes \wedge^\bullet\fg^*$ of such maps
becomes a complex under the Chevalley--Eilenberg differential
$\d : C^k(\fg,\fM) \to C^{k+1}(\fg,\fM)$ defined by extending linearly
the linear map
\begin{equation}
  \label{eq:ce-diff}
  \d (m \otimes \omega) = \d m \wedge \omega + m \otimes \d \omega,
\end{equation}
for $m \in \fM$ and $\omega \in \wedge^k\fg^*$, where $\d\omega$ is
induced, as before, by the de~Rham differential and $\d m : \fg \to \fM$
is defined by $\d m(X) = \rho(X) \, m$ for all $X \in \fg$.  The
cohomology of this complex is denoted $H^\bullet(\fg,\fM)$ and is called
the Lie algebra cohomology of $\fg$ with values in the module $\fM$.

The preceding discussion extends to the case of $\fg$ a Lie superalgebra
and $\fM$ a module, except that now $\wedge$ is taken in the
super-sense (i.e., symmetric on odd elements) and we must insert the
relevant Koszul signs.  Explicitly, the Chevalley--Eilenberg
differential $\d: C^k(\fg,\fM) \to C^{k+1}(\fg,\fM)$ in low degree is
given by
\begin{align}
  \begin{split}\label{eq:Spencer0}
    &\d : C^0(\fg,\fM) \to C^1 (\fg,\fM)\\
    &\d m (X) = \rho(X)\, m~,
  \end{split}
  \\
  \begin{split}\label{eq:Spencer1}
    &\d : C^1(\fg,\fM)\to C^2(\fg,\fM)\\
    &\d\varphi(X,Y) = \rho(X) \, \varphi(Y) - (-1)^{xy} \rho(Y) \, \varphi(X) - \varphi([X,Y])~,
  \end{split}
\\
  \begin{split}\label{eq:Spencer2}
    &\d : C^2(\fg,\fM) \to C^3(\fg,\fM)\\
    &\d\varphi(X,Y,Z) = \rho(X)\, \varphi(Y,Z) + (-1)^{x(y+z)} \rho(Y)\,
    \varphi(Z,X) + (-1)^{z(x+y)} \rho(Z)\, \varphi(X,Y) \\
    & {} \qquad\qquad\qquad - \varphi([X,Y],Z) - (-1)^{x(y+z)} \varphi([Y,Z],X) -(-1)^{z(x+y)} \varphi([Z,X],Y)~,
  \end{split}
\end{align}
where $x,y,z$ denote the parity of elements $X,Y,Z$ of $\fg$.

If we take $\fM = \fg$ to be the adjoint representation, with $\rho(X)\,
Y= [X,Y]$, then $H^\bullet(\fg,\fg)$ has the following
interpretation in low degree:
\begin{itemize}
\item $H^0(\fg,\fg)$ is isomorphic to the centre of $\fg$;
\item $H^1(\fg,\fg)$ is the quotient of the derivations of $\fg$ by the
  inner derivations; and
\item $H^2(\fg,\fg)$ are the infinitesimal deformations of $\fg$.
\end{itemize}
In this paper we concentrate on a refinement of this latter cohomology
group for the case of graded Lie superalgebras.  Rather than discuss
this in full generality, let us already specialise to the Poincaré
superalgebra defined in Section~\ref{sec:poin-super}.

\subsection{The Spencer complex}
\label{sec:spencer-complex}

Let $\fs_- := \fs_{-1} \oplus \fs_{-2}$ denote the negative-degree
subalgebra of the Poincaré superalgebra.  It is a Lie superalgebra in
its own right and the restriction of the adjoint action of $\fs$ to
$\fs_-$ makes $\fs$ into a module of $\fs_-$.  We define
$C^p(\fs_-;\fs) = \Hom(\wedge^p\fs_-,\fs)$, where $\wedge^p$ is taken in
the super-sense.  The direct sum $C^\bullet := \bigoplus_p C^p(\fs_-;\fs)$
becomes a differential complex relative to the Chevalley--Eilenberg
differential of $\fs_-$ relative to its module $\fs$.  We extend the
$\ZZ$-grading of $\fs$ to $C^\bullet$ in the natural way and, since
$\fs$ is $\ZZ$-graded, the differential has degree $0$.  Therefore the
Chevalley--Eilenberg complex decomposes into subcomplexes of a fixed
degree: we let $C^{d,p}(\fs_-;\fs)$ denote the $p$-cochains of degree
$d$ and $\partial : C^{d,p}(\fs_-;\fs) \to C^{d,p+1}(\fs_-;\fs)$ denote the
restriction of the Chevalley--Eilenberg differential.  We are interested
in calculating $H^{2,2}(\fs_-;\fs)$.  The relevant spaces of cochains
are:
\begin{equation}
  \begin{split}
    C^{2,1}(\fs_-;\fs) &= \Hom(V,\fso(V))\\
    C^{2,2}(\fs_-;\fs) &= \Hom(\wedge^2 V, V) \oplus \Hom(V \otimes S,
    S) \oplus \Hom(\odot^2 S, \fso(V))\\
    C^{2,3}(\fs_-;\fs) &= \Hom(\odot^2S \otimes V, V) \oplus \Hom(\odot^3 S, S),
  \end{split}
\end{equation}
and the differential can be read off from equations~\eqref{eq:Spencer1}
and \eqref{eq:Spencer2} above, with $\rho(X) Y = [X,Y]$ for $X \in
\fs_-$ and $Y \in \fs$.

As usual, the Spencer differential $\partial : C^{2,1}(\fs_-;\fs) \to
C^{2,2}(\fs_-;\fs)$ is injective and hence $H^{2,1}(\fs_-;\fs) = 0$.
Moreover, $H^{2,2}(\fs_-;\fs) \cong \eH^{2,2}$, where
\begin{equation}
  \eH^{2,2} = \ker \pi_1 \cap Z^{2,2},
\end{equation}
where $Z^{2,2}$ is the space of Spencer cocycles in $C^{2,2}(\fs_-;\fs)$
and $\pi_1 : C^{2,2}(\fs_-;\fs) \to \Hom(\wedge^2 V,V)$ is the natural
projection onto the first summand.  In other words, cohomology classes
in $H^{2,2}(\fs_-;\fs)$ are in bijective correspondence with normalised
cocycles $\beta + \gamma \in Z^{2,2}$, where $\beta : V \otimes S \to S$
and $\gamma : \odot^2S \to \fso(V)$.

There are two components to the cocycle condition: one in $\Hom(\odot^2S
\otimes V, V)$:
\begin{equation}
  \label{eq:cocycle-1}
  \gamma(s,s)  v + 2 [s, \beta(v,s)] = 0,
\end{equation}
and one in $\Hom(\odot^3 S, S)$
\begin{equation}
  \label{eq:cocycle-2}
  \gamma(s,s) s + \beta([s,s],s) = 0,
\end{equation}
for all $v \in V$ and $s \in S$.

\subsection{Solution of the first cocycle condition}
\label{sec:solut-first-cocycle}

We define $\beta_v \in \End(S)$ by $\beta_v(s) = \beta(v,s)$ for $v \in
V$ and $s \in S$.  We let $\beta_\mu := \beta_{\be_\mu}$ and parametrise $\beta_\mu$ as follows:
\begin{equation}
  (\beta_\mu s)^B = A_\mu s^B + B_\mu{}^B{}_C s^C + C_{\mu\nu}
  \Gamma^\nu s^B + D_{\mu\nu}{}^B{}_C \Gamma^\nu s^C + \half
  E_{\mu\nu\rho} \Gamma^{\nu\rho} s^B + \half F_{\mu\nu\rho}{}^B{}_C
  \Gamma^{\nu\rho} s^C,
\end{equation}
where \(B_\mu,D_{\mu\nu},F_{\mu\nu\rho}\in\fsp(1)\). In particular, with all lowered \(\Delta\)-indices, these components are symmetric in those indices. The first cocycle condition~\eqref{eq:cocycle-1} becomes
\begin{equation}
  \gamma(s,s)_{\mu\nu} + 2 \epsilon_{AB} \sbar^A\Gamma_\mu (\beta_\nu
  s)^B= 0.
\end{equation}
The skew-symmetric part in $\mu\nu$ expresses $\gamma(s,s)$ in terms of $\beta$,
whereas the symmetric part constrains $\beta$.

\begin{lemma}[First cocycle condition]\label{lem:cocycle-1}
  The solution of the first cocycle condition~\eqref{eq:cocycle-1} is
  given by
  \begin{equation}
    \begin{split}
      (\beta_\mu s)^B &= B_\mu{}^B{}_C s^C + C_{\mu\nu} \Gamma^\nu s^B +  \tfrac15 D^B{}_C \Gamma_\mu s^C + \half E_{\mu\nu\rho} \Gamma^{\nu\rho} s^B + \tfrac14 F^\rho{}^B{}_C \Gamma_{\mu\rho} s^C\\
      \gamma(s,s)_{\mu\nu} &= 2 \mu C_{\mu\nu} - \tfrac25 D_{AB}
      \omega^{AB}_{\mu\nu} + 2 E_{\mu\nu\rho} \kappa^\rho + \tfrac14
      \epsilon_{\mu\nu\rho}{}^{\sigma\tau} F^\rho_{AB} \omega^{AB}_{\sigma\tau},
    \end{split}
  \end{equation}
  for some $C \in \wedge^2V$ and $E \in \wedge^3 V$.
\end{lemma}

\begin{proof}
  Symmetrising the first cocycle condition in $\mu\nu$ results in
  \begin{equation}
    0 = \epsilon_{AB} \sbar^A \left( \Gamma_\mu (\beta_\nu  s)^B  +
      \Gamma_\nu (\beta_\mu  s)^B \right),
  \end{equation}
  which expands to
  \begin{multline}
    0 = 2 \kappa_{(\mu} A_{\nu)} + 2 \mu C_{(\mu\nu)} +
    D_\mu{}^\rho{}_{AB} \omega_{\nu\rho}^{AB} +  D_\nu{}^\rho{}_{AB}
    \omega_{\mu\rho}^{AB} + 2 E_{(\mu\nu)\rho} \kappa^\rho\\ - \tfrac14 \epsilon_{\mu\rho\sigma}{}^{\tau\xi}
    F_\nu{}^{\rho\sigma}{}_{AB} \omega_{\tau\xi}^{AB} - \tfrac14 \epsilon_{\nu\rho\sigma}{}^{\tau\xi}
    F_\mu{}^{\rho\sigma}{}_{AB} \omega_{\tau\xi}^{AB}.
  \end{multline}
  Since this equation holds for all $s\in S$, the terms in $\mu$,
  $\kappa^\mu$ and $\omega_{\mu\nu}^{AB}$ must separately vanish.  The
  $\mu$-term simply says that $C_{(\mu\nu)} = 0$, so that $C \in
  \wedge^2 V$.  (We do not distinguish between $V$ and $V^*$.)  The
  $\kappa$-terms result in the equation
  \begin{equation}
    0 = E_{\mu\nu\rho} + E_{\nu\mu\rho} + \eta_{\mu\rho} A_\nu +
    \eta_{\nu\rho} A_\mu.
  \end{equation}
  Adding the cyclic permutations of this equation and using that
  $E_{\mu\nu\rho} = - E_{\mu\rho\nu}$, we find that
  \begin{equation}
    \eta_{(\mu\nu} A_{\rho)} = 0 \implies A_\mu = 0,
  \end{equation}
  and hence that $E_{\mu\nu\rho} = - E_{\nu\mu\rho}$, which says that $E
  \in \wedge^3 V$.  Finally, the $\omega$-terms result in
  \begin{equation}\label{eq:fcc-omega}
    0 = D_{\mu\xi} \eta_{\tau\nu} + D_{\nu\xi} \eta_{\tau\mu} -
    D_{\mu\tau} \eta_{\xi\nu} - D_{\nu\tau} \eta_{\xi\mu}- \tfrac12
    \epsilon_{\mu\rho\sigma\tau\xi} F_\nu{}^{\rho\sigma} - \tfrac12
    \epsilon_{\nu\rho\sigma\tau\xi} F_\mu{}^{\rho\sigma},
  \end{equation}
  where we have omitted the $\fsp(1)$-indices.  Tracing with
  $\eta^{\mu\nu}$ we find
  \begin{equation}\label{eq:DF-1}
    D_{[\tau\xi]} = \tfrac14 \epsilon_{\tau\xi\nu\rho\sigma} F^{\nu\rho\sigma},
  \end{equation}
  whereas tracing with $\eta^{\nu\tau}$ yields
  \begin{equation}
    5 D_{\mu\xi} - D \eta_{\mu\xi} + \tfrac12
    \epsilon_{\mu\xi\rho\sigma\tau} F^{\tau\rho\sigma} = 0,
  \end{equation}
  where we have introduced $D := \eta^{\mu\nu} D_{\mu\nu}$.  Breaking up
  into symmetric and skew-symmetric parts, we find
  \begin{equation}
    D_{(\mu\xi)} = \tfrac15 \eta_{\mu\xi} D \qquad\text{and}\qquad
    D_{[\mu\xi]} = -\tfrac1{10} \epsilon_{\mu\xi\rho\sigma\tau} F^{\rho\sigma\tau}.
  \end{equation}
  Taking these equations together with equation~\eqref{eq:DF-1}, we
  conclude that $D_{\mu\nu} = \tfrac15 \eta_{\mu\nu} D$ and
  $\epsilon_{\mu\nu\rho\sigma\tau} F^{\rho\sigma\tau} = 0$, so that the
  $\wedge^3V$ component of $F$ vanishes.  Back into the
  $\omega$-equation~\eqref{eq:fcc-omega}, we find
  \begin{equation}
    \epsilon_{\mu\rho\sigma\tau\xi} F_\nu{}^{\rho\sigma} +
    \epsilon_{\nu\rho\sigma\tau\xi} F_\mu{}^{\rho\sigma} = 0.
  \end{equation}
  Contracting with $\epsilon^{\nu\tau\xi\alpha\beta}$, we arrive at
  \begin{equation}
  	16 F_\mu{}^{\alpha\beta} + 4 \delta_\mu^\beta F_\nu{}^{\nu\alpha} - 4 \delta_\mu^\alpha F_\nu{}^{\nu\beta} = 0
  	\implies
    F_{\mu\nu\rho} = \tfrac14 \left(\eta_{\mu\nu} F_\rho -
      \eta_{\mu\rho} F_\nu\right),
  \end{equation}
  where we have defined $F_\mu := F^\nu{}_{\nu\mu}$.  Collecting these
  results together, we arrive at the expressions for $\beta$ and
  $\gamma$ in the statement of the lemma.
\end{proof}

\subsection{Solution of the second cocycle condition}
\label{sec:solut-second-cocycle}

The second cocycle condition~\eqref{eq:cocycle-2} becomes
\begin{equation}\label{eq:scc-2}
  \tfrac14 \gamma(s,s)_{\mu\nu}\Gamma^{\mu\nu} s^B + \kappa^\mu
  (\beta_\mu s)^B = 0,
\end{equation}
where $\gamma(s,s)_{\mu\nu}$ and $(\beta_\mu s)^B$ are given in
Lemma~\ref{lem:cocycle-1}.

\begin{lemma}[Second cocycle condition]
  The second cocycle condition reduces to
  \begin{equation}
    D^B{}_C = 0, \qquad B_\mu{}^B{}_C = - \half F_\mu{}^B{}_C
    \qquad\text{and}\qquad E_{\mu\nu\rho} = \tfrac14
    \epsilon_{\mu\nu\rho\sigma\tau} C^{\sigma\tau}.
  \end{equation}
\end{lemma}

\begin{proof}
  From the expressions for $\beta$ and $\gamma$ in
  Lemma~\ref{lem:cocycle-1}, the second cocycle condition
  \eqref{eq:scc-2} can be written as
  \begin{equation}\label{eq:scc-short}
    \left(\mu \Theta^B{}_C + \kappa^\mu \Psi_\mu{}^B{}_C + \tfrac12
    \omega_{\mu\nu}^{AD} \Phi_{AD}^{\mu\nu} \delta^B_C\right) s^C = 0,
  \end{equation} 
  where we have introduced the shorthands
  \begin{equation}
    \begin{split}
      \Theta^B{}_C &:= \tfrac12 C_{\mu\nu}\Gamma^{\mu\nu} \delta^B_C + \tfrac12 D^B{}_C
      - \tfrac14 F_\mu{}^B{}_C \Gamma^\mu\\
      \Psi_\mu{}^B{}_C &:= E_{\mu\nu\rho}\Gamma^{\nu\rho} \delta^B_C +
      B_\mu{}^B{}_C + \tfrac14 F_\mu{}^B{}_C + C_{\mu\nu} \Gamma^\nu
      \delta^B_C\\
      \Phi_{AB}^{\mu\nu} &:= \tfrac18 \epsilon^{\mu\nu\rho\sigma\tau}  F_{\rho\,AB} \Gamma_{\sigma\tau}.
    \end{split}
  \end{equation}
  We expand $\mu$, $\kappa$ and $\omega$ in
  equation~\eqref{eq:scc-short} to obtain
  \begin{equation}
    \left(\epsilon_{AB} \sbar^A s^B \Theta^C{}_D + \epsilon_{AB} \sbar^A
    \Gamma^\mu s^B \Psi_\mu{}^C{}_D + \tfrac12 \sbar^A \Gamma_{\mu\nu}
    s^B \Phi_{AB}^{\mu\nu} \delta^C_D\right) s^D = 0
  \end{equation}
  and we polarise away from the diagonal:
  \begin{equation}
    \begin{split}
      & \left(\epsilon_{AB} \sbar_1^A s_2^B \Theta^C{}_D + \epsilon_{AB}
        \sbar_1^A  \Gamma^\mu s_2^B \Psi_\mu{}^C{}_D + \tfrac12
        \sbar_1^A \Gamma_{\mu\nu} s_2^B \Phi_{AB}^{\mu\nu}
        \delta^C_D\right) s_3^D\\
      + {} & \left(\epsilon_{AB} \sbar_2^A s_3^B \Theta^C{}_D + \epsilon_{AB}
        \sbar_2^A  \Gamma^\mu s_3^B \Psi_\mu{}^C{}_D + \tfrac12
        \sbar_2^A \Gamma_{\mu\nu} s_3^B \Phi_{AB}^{\mu\nu}
        \delta^C_D\right) s_1^D\\
      + {} & \left(\epsilon_{AB} \sbar_1^A s_3^B \Theta^C{}_D + \epsilon_{AB}
        \sbar_1^A  \Gamma^\mu s_3^B \Psi_\mu{}^C{}_D + \tfrac12
        \sbar_1^A \Gamma_{\mu\nu} s_3^B \Phi_{AB}^{\mu\nu}
        \delta^C_D\right) s_2^D = 0,
    \end{split}
  \end{equation}
  where in the last line we have used the symmetry properties of
  Lemma~\ref{sec:idents}.  We write this as an endomorphism
  acting on $s_3^D$:
  \begin{multline}
    \left(\epsilon_{AB} \sbar_1^A s_2^B \Theta^C{}_D + \epsilon_{AB}
      \sbar_1^A  \Gamma^\mu s_2^B \Psi_\mu{}^C{}_D + \tfrac12
      \sbar_1^A \Gamma_{\mu\nu} s_2^B \Phi_{AB}^{\mu\nu} \delta^C_D
    \right. \\
    + \left. 2 \epsilon_{AD} \Theta^C{}_B s_1^B \sbar_2^A s_3^D + 2
      \epsilon_{AD} \Psi_\mu{}^C{}_B s_1^B \sbar_2^A  \Gamma^\mu +  \Phi_{AD}^{\mu\nu}
      s_1^C \sbar_2^A \Gamma_{\mu\nu}\right) s_3^D = 0.
  \end{multline}
  This is true for all $s_3$, hence the endomorphism in parenthesis has
  to vanish for all $s_1$ and $s_2$.  Being symmetric in $s_1$ and
  $s_2$, it is uniquely characterised by its values on the diagonal, so
  we can simply take $s_1 = s_2 = s$ and rewrite the second cocycle
  condition as
  \begin{multline}
    \epsilon_{AB} \sbar^A s^B \Theta^C{}_D + \epsilon_{AB}
    \sbar^A  \Gamma^\mu s^B \Psi_\mu{}^C{}_D + \tfrac12
    \sbar^A \Gamma_{\mu\nu} s^B \Phi_{AB}^{\mu\nu} \delta^C_D \\
    + 2 \epsilon_{AD} \Theta^C{}_B s^B \sbar^A s_3^D + 2
    \epsilon_{AD} \Psi_\mu{}^C{}_B s^B \sbar^A  \Gamma^\mu +  \Phi_{AD}^{\mu\nu}
    s^C \sbar^A \Gamma_{\mu\nu} = 0.
  \end{multline}
  We use the Fierz identity~\eqref{eq:fierztoo} in the last three terms
  and we collect terms in $\mu$, $\kappa$ and $\omega$ together, each of
  which must vanish (as an endomorphism of $S$) separately, resulting in
  three equations:
  \begin{equation}\label{eq:scc-mu}
    5 \Theta^C{}_D + \Psi_\mu{}^C{}_D \Gamma^\mu + \tfrac12
    \epsilon^{AC} \Phi_{AD}^{\mu\nu}\Gamma_{\mu\nu} = 0,
  \end{equation}
  \begin{equation}\label{eq:scc-kappa}
    5 \Psi_\mu{}^C{}_D + \Theta^C{}_D \Gamma_\mu + \Psi^{\nu\,C}{}_D
    \Gamma_{\mu\nu} + \tfrac12 \epsilon^{AC}\Phi_{AD}^{\nu\rho}
    \Gamma_\mu \Gamma_{\nu\rho} = 0,
  \end{equation}
  and
  \begin{equation}\label{eq:scc-omega}
    \Phi^{\mu\nu}_{AB} \delta^C_D + \tfrac12
    \epsilon_{D(A}\Theta^C{}_{B)} \Gamma^{\mu\nu} - \tfrac12
    \Psi_\rho{}^C{}_{(A} \epsilon_{B)D} \Gamma^{\mu\nu} \Gamma^\rho -
    \tfrac14 \delta^C_{(A} \Phi^{\rho\sigma}_{B)D}
    \Gamma^{\mu\nu}\Gamma_{\rho\sigma} = 0.
  \end{equation}
  We notice that the $\mu$-equation~\eqref{eq:scc-mu} is the
  Clifford trace of the $\kappa$-equation~\eqref{eq:scc-kappa}: indeed,
  contracting equation~\eqref{eq:scc-kappa} on the right with
  $\Gamma^\mu$ we obtain equation~\eqref{eq:scc-mu}.  Therefore the
  $\mu$-equation is redundant and we concentrate on the other two
  equations.
  
  The $\kappa$-equation can be rewritten as
  \begin{equation}
    5 \Psi_{\mu\,AB} + \Theta_{AB}\Gamma_\mu + \Psi^\nu_{AB}
    \Gamma_{\mu\nu} + \half \Phi^{\nu\rho}_{AB}\Gamma_\mu
    \Gamma_{\nu\rho} = 0.
  \end{equation}
  Upon substituting the expressions for $\Theta$, $\Psi$ and $\Phi$ into
  this equation, we first symmetrise in $AB$ (dropping the indices) and
  using the Clifford relations we arrive at
  \begin{equation}
    5 (B_\mu + \half F_\mu) + (B^\nu + \half F^\nu)\Gamma_{\mu\nu} +
    \half D \Gamma_\mu = 0,
  \end{equation}
  which results in
  \begin{equation}
    D_{AB} = 0 \qquad\text{and}\qquad B^\mu_{AB} = - \half F^\mu_{AB}.
  \end{equation}
  If instead we skew-symmetrise in $AB$ we arrive at
  \begin{equation}\label{eq:scc-kappa-skew}
    5 (E_{\mu\nu\rho} \Gamma^{\nu\rho} + C_{\mu\nu} \Gamma^\nu) + \half
    C_{\nu\rho}\Gamma^{\nu\rho}\Gamma_\mu +
    (E^{\nu\sigma\tau}\Gamma_{\sigma\tau} + C^{\nu\rho}_\rho)
    \Gamma_{\mu\nu} = 0.
  \end{equation}
  Contracting with $\Gamma^\mu$ on the right and simplifying we find
  \begin{equation}
    3 C_{\mu\nu} = \epsilon_{\mu\nu\rho\sigma\tau} E^{\rho\sigma\tau},
  \end{equation}
  which can be inverted to write
  \begin{equation}
    E_{\mu\nu\rho} = \tfrac14 \epsilon_{\mu\nu\rho\sigma\tau}
    C^{\sigma\tau}.
  \end{equation}
  Re-inserting this back into equation~\eqref{eq:scc-kappa-skew}, we
  find that it is identically satisfied.

  One can check that the remaining equation~\eqref{eq:scc-omega} is
  identically satisfied.
\end{proof}

In summary, we have proved the following

\begin{theorem}\label{thm:Spencer-calculation}
  As a module of $\fso(V) \oplus \fsp(1)$, we have the following
  isomorphism:
  \begin{equation}
    H^{2,2}(\fs_-;\fs) \cong \left(\wedge^2 V \otimes \wedge^2 \Delta
    \right) \oplus  \left( V \otimes \odot^2\Delta\right),
  \end{equation}
  where to a class $(C_{\mu\nu}\epsilon_{AB}, F^\mu_{AB}) \in H^{2,2}$
  there corresponds the cocycle $(\beta, \gamma) \in Hom(V\otimes S,
  S) \oplus \Hom(\odot^2S, \fso(V))$ given by
  \begin{equation}\label{eq:spencer-22-cocycle-solution}
    \begin{split}
      \beta(v,s)^B &=  \tfrac14 v \cdot C \cdot s^B - \tfrac34 C \cdot v
      \cdot s^B - \tfrac18 v \cdot F^B{}_C\cdot s^C - \tfrac38 F^B{}_C \cdot v
      \cdot s^C\\
      \gamma(s,s)_{\mu\nu} &= 2 \mu C_{\mu\nu} + \tfrac12 \kappa^\rho
      \epsilon_{\mu\nu\rho\sigma\tau} C^{\sigma\tau} + \tfrac14
      \epsilon_{\mu\nu\rho}{}^{\sigma\tau} F^\rho_{AB} \omega^{AB}_{\sigma\tau}.
    \end{split}
  \end{equation}
\end{theorem}

This is not far from the naive dimensional reduction of the result in
\cite{deMedeiros:2018ooy}: the self-dual $3$-form reduces dimensionally
to a $2$-form ($C_{\mu\nu}$) and a $3$-form ($E_{\mu\nu\rho}$) which are
related by Hodge duality and the $\fsp(1)$-valued $1$-form gives rise to
an $\fsp(1)$-valued $1$-form ($F_\mu{}^A{}_B$) and an $\fsp(1)$-valued
scalar, which is missing from the five-dimensional calculation.  The
precise behaviour of the Spencer cohomology under dimensional reduction
lies beyond the scope of this paper and will be addressed in a separate
paper.

\section{Zero-curvature equations}
\label{sec:zero-curv-eq}

In this section we interpret the cohomological calculations of the
previous section geometrically.  The first step is to re-interpret the
Spencer complex geometrically and we do this in
Section~\ref{sec:setup-geometry}, arriving at a connection $\eD$ on
spinor fields, whose curvature we calculate in
Section~\ref{sec:curvature-calc}.  The final aim of this section is to
derive the conditions for maximal supersymmetry, which at least locally
is tantamount to the flatness of $\eD$.  We do this in two steps: in
Section~\ref{sec:gamma-trace} we impose the vanishing of the Clifford
trace of the curvature, which in the supergravity context often
coincides with the bosonic field equations, and finally in
Section~\ref{sec:flat-connection} we derive the conditions for vanishing
curvature.  In Section~\ref{sec:killing-superalg} we will show that the
$\eD$-parallel spinor fields generate a Lie superalgebra at least when
the curvature of $\eD$ is Clifford-traceless and in
Section~\ref{sec:max-susy} we will study the geometries on which $\eD$
is flat.

\subsection{Setup}
\label{sec:setup-geometry}

We shall fix a five-dimensional lorentzian spin manifold $(M,g)$ and let
$V$ be as in Section~\ref{sec:poinc-super}; that is, $V$ is a
five-dimensional lorentzian vector space we may identify with
$\RR^{1,4}$.  The spin bundle $\Spin(M)$ is a principal
$\Spin(V)$-bundle which comes with a bundle morphism
$\Spin(M) \to \SO(M)$ to the oriented orthonormal frame bundle, covering
the identity and agreeing fibrewise with the standard 2-to-1 covering
homomorphism $\Spin(V) \to \SO(V)$.  The principal bundle $\SO(M)$ is a
$G$-structure with $G=\SO(V)$ and therefore comes with a soldering form
which restricts pointwise to a vector space isomorphism $T_p M \to V$.
These isomorphisms assemble to a bundle isomorphism between $TM$ and the
``fake tangent bundle'' $\SO(M) \times_{\SO(V)} V$, which is the
associated vector bundle of $\SO(M)$ corresponding to the vector
representation of $\SO(V)$. Using this construction, we may locally write
the components of tensor fields on \(M\) as if they were tensors on \(V\)
using the orthonormal basis on \(V\). This may be equivalently viewed as 
working in a local orthonormal frame on \(M\).

Since we are interested in spin manifolds, the relevant principal bundle
is $\Spin(M)$.  Any associated vector bundle to $\SO(M)$ can be
interpreted as an associated vector bundle to $\Spin(M)$ via the bundle
morphism $\Spin(M) \to \SO(M)$, but there are of course also associated
vector bundles to $\Spin(M)$ which do not arise in this way: namely,
those involving spinorial representations.

As in Section~\ref{sec:spin-conv}, let $\Sigma$ denote one of the two
inequivalent Clifford modules of $\Cl(V)$.  It becomes a
$\Spin(V)$-module by restriction.  Let
$\$ := \Spin(M) \times_{\Spin(V)} \Sigma$ denote the corresponding
spinor bundle.  It is a complex rank-4 vector bundle with a quaternionic
structure $J$ and an invariant symplectic inner product
$(\sigma,\tau) \mapsto \bar\sigma\tau$ inherited from $\Sigma$.  We
introduce an auxiliary trivial\footnote{It could be interesting to relax
  this condition, gauge the R-symmetry and introduce a connection on a
  possibly non-trivial R-symmetry bundle, but that lies beyond the scope
  of the present paper.} complex rank-2 vector bundle
$\eH \cong M \times \Delta$.  It too has a quaternionic structure $j$
and a symplectic inner product $(\cdot,\cdot)$ inherited from $\Delta$.
We will make a global choice of symplectic frame $e_1,e_2$ for $\eH$
such that $(e_A,e_B) = \epsilon_{AB}$, where $\epsilon_{AB}$ is as in
Section~\ref{sec:spin-conv}.  On the tensor product bundle
$\bbS_\CC := \$ \otimes \eH$ we have an invariant real structure
$J \otimes j$.  Its real sub-bundle $\bbS$ is the real rank-8 vector
bundle associated to $\Spin(M)$ via the representation $S$, whose
complexification $S \otimes_\RR \CC = \Sigma \otimes_\CC \Delta$.  Any
spinor section $s$ of $\bbS_\CC$ or $\bbS$ may be expanded relative to
the global frame $e_A$ as $s = s^A \otimes e_A$, where in the case of
$\bbS$ the $s^A$ are subject to the reality condition
\eqref{eq:symplectic-majorana}.  The symplectic inner products on $\$$
and $\eH$ combine to an inner product on
$\bbS_\CC$ given by $\left<s_1,s_2\right> = \epsilon_{AB}\bar s_1^A
s_2^B$, which is real when restricted to
$\bbS$.  We shall refer to sections of
$\bbS_\CC$ as \textbf{(symplectic) Dirac spinor fields} and to sections of
$\bbS$ as \textbf{(symplectic) Majorana spinor fields}.

In summary, with any representation $W$ of $\Spin(V)$ made out of
$V$ and $S$ (via tensor product and taking duals) we can associate a
vector bundle $\Spin(M) \times_{\Spin(V)} W$ in such a way that to any
$\Spin(V)$-equivariant linear map $\varphi : W_1 \to W_2$ between two
such representations, we associate a corresponding bundle morphism
$\Spin(M) \times_{\Spin(V)} W_1 \to \Spin(M) \times_{\Spin(V)} W_2$.
Since the cochains in the generalised Spencer complex are
$\Spin(V)$-modules and the differential is $\Spin(V)$-equivariant, we
may interpret the Spencer complex as a complex of
the associated vector bundles and, in particular, as a complex on their
spaces of smooth sections, and similarly for its cohomology.

In this way, for example, the component $\beta$ of the Spencer cocycle
in equation~\eqref{eq:spencer-22-cocycle-solution} can be interpreted as
a section of the vector bundle associated to the representation $\Hom(V,
\End(S))$, which is the bundle of one-forms with values in
$\End(\bbS)$.  This is the bundle on whose space of sections the affine
space of connections on $\bbS$ is modelled on and therefore we
may understand $\beta$ as the difference between two such connections.
The natural connection on $\bbS$ is the spin connection $\nabla$ -- that
is, the one induced from the lift to $\Spin(M)$ of the Ehresmann connection on
$\SO(M)$ which induces the Levi-Civita connection on $TM$ -- and
therefore we may interpret the cocycle component $\beta$ as $\nabla -
\eD$ for some connection $\eD$ on $\bbS$.  The cocycle component $\beta$
depends on the additional geometric data parametrising the relevant
Spencer cohomology, namely a two-form $C \in \Omega^2(M)$ and an
$\fsp(1)$-valued one form $F \in \Omega^1(M,\fsp(1))$.
 
Thus let us define a \textbf{(bosonic) background} \((M,g,C,F)\) to be a spin
manifold \((M,g)\) with the structures described above along with the
forms \(C\in\Omega^2(M)\) and \(F\in\Omega^1(M;\fsp(1))\). It
will be useful in places to view the components \(F\indices{^A_B}\) of
\(F\) in a symplectic frame for \(\eH\) as 1-forms on \(M\).  The
$\End(\bbS)$-valued one-form $\beta$ corresponding to the Spencer
cocycle in equation~\eqref{eq:spencer-22-cocycle-solution} is given by
\begin{equation}\label{eq:beta-clifford}
  \beta_X s = \tfrac{1}{4}X\cdot C\cdot s - \tfrac{3}{4}C\cdot X\cdot s
  - \tfrac{1}{8}X\cdot F\cdot s - \tfrac{3}{8}F\cdot X\cdot s,
\end{equation} 
where \(s\) is a Majorana spinor field, \(X\) is a vector
field, \(\cdot\) denotes both the Clifford multiplication of forms and
the Clifford action of forms on spinor fields, and \(F\) also acts via
the \(\fsp(1)\) action on \(\Gamma(\bbS)\) to the right.\footnote{We abuse
notation in that on the right-hand side we should have the metrically
dual one-form $X^\flat$ and not the vector field, but we trust this
ought not be a cause of confusion.}

With the conventions chosen above, in components we have
\begin{equation}\label{eq:beta-comps} (\beta_\mu s)^A =
  \tfrac{1}{8}C^{\alpha\beta}\qty(\epsilon_{\mu\alpha\beta\sigma\tau}\Gamma^{\sigma\tau}
  + 8\eta_{\mu\alpha}\Gamma\indices{_\beta}) s^A +
  \tfrac{1}{4}F\indices{^\alpha^A_B}\qty(\Gamma\indices{_{\mu\alpha}} -
  2\eta_{\mu\alpha})s^B.
 \end{equation}
We define the \textbf{superconnection} \(\eD\) on \(\bbS\) by
\begin{equation}\label{eq:superconnection} \eD_X s = \nabla_X s -
  \beta_X s,
\end{equation} where \(\nabla\) is (the spin lift of) the
Levi-Civita connection, \(X\in\fX(M)\) and \(s\in\Gamma(\bbS)\). In
components,
\begin{equation} (\eD_\mu s)^A = \nabla_\mu s^A -
  \tfrac{1}{8}C^{\alpha\beta}\qty(\epsilon_{\mu\alpha\beta\sigma\tau}\Gamma^{\sigma\tau}
  + 8\eta_{\mu\alpha}\Gamma\indices{_\beta}) s^A -
  \tfrac{1}{4}F\indices{^\alpha^A_B}\qty(\Gamma\indices{_{\mu\alpha}} -
  2\eta_{\mu\alpha})s^B.
\end{equation}

Let us remark that, contrary to what one might have suspected, the
$\fsp(1)$-valued one-form $F$ does not correspond to the difference
between two connections on the auxiliary bundle $\eH$.  If that were the
case, the term $\tfrac{1}{4}F\indices{^\alpha^A_B}
\Gamma\indices{_{\mu\alpha}}$ in the expression for $\eD$ would be
absent, while its presence suggests a mixing of local Lorentz and
R-symmetries.

\begin{definition}
  A \textbf{Killing spinor (field)} on a background \((M,g,C,F)\) is a
  spinor field \(s\in\Gamma(\bbS)\) which is parallel with respect to
  the superconnection \(\eD\); that is, if it satisfies the
  \textbf{Killing spinor equation}
  \begin{equation}
    \nabla s = \beta s.
  \end{equation}
  A background \((M,g,C,F)\) is \textbf{supersymmetric} if it admits a
  Killing spinor, and \textbf{maximally supersymmetric} if its space of
  Killing spinors has maximal dimension.
\end{definition}

The notion of "maximal dimension" in the definition above arises because
a set of linearly independent sections has linearly independent values
at all points; hence such a set has size at most \(\rank\bbS=\dim S\).
The following proposition gives a necessary condition for maximal
supersymmetry.

\begin{proposition}\label{prop:max-susy-vanishing-curvature}
  If a background \((M,g,C,F)\) is maximally supersymmetric, it is flat
  with respect to the superconnection: the curvature tensor \({R^\eD}\),
  given by
  \begin{equation}\label{eq:Dcurv}
    {R^\eD}(X,Y) = \eD_{\comm{X}{Y}} - \comm{\eD_X}{\eD_Y}
  \end{equation}
  where \(X,Y\) are vector fields, vanishes. The converse holds if \(M\)
  is simply connected.
\end{proposition}

\begin{proof}
  Clearly from the definition, \(R^\eD\) annihilates Killing spinors. If
  there are \(\rank\bbS=\dim S\) linearly independent Killing spinors,
  at any point \(x\in M\) their values span the fibre \(\bbS_x\), thus
  \(R^\eD\) must annihilate all spinors at \(x\), hence \(R^\eD\)
  vanishes.

  Conversely, assume \(M\) is simply connected and \(\eD\) is flat. Any
  choice of a spinor at any point determines a Killing spinor by
  parallel transport, and furthermore, parallel transport of a basis at
  any point determines \(\dim S\) linearly independent Killing spinors.
\end{proof}

Determining the curvature tensor will allow necessary (and sufficient,
in the case of simply-connected backgrounds) conditions for maximal
supersymmetry to be found. It is also possible to recover the bosonic
supergravity equations by imposing a weaker restriction than the
vanishing of the curvature, namely the vanishing of its
Clifford trace: \(\Gamma^\nu {R^\eD}_{\mu\nu} =0\). Indeed, after
finding \({R^\eD}\), our approach will be to calculate its
Clifford trace and derive necessary and sufficient conditions for this
to vanish, which will simplify the vanishing curvature calculation.

\subsection{Conventions on curvature tensors}
\label{sec:curvature-convention}

We define the curvature \(R^D\in\Omega^2(M;\End(E))\) of any connection \(D\) on a vector bundle \(E\) over \(M\) by
\begin{equation}
	R^D(X,Y) = D_{\comm{X}{Y}}-\comm{D_X}{D_Y}.
\end{equation}
In particular, \(R^\nabla\) is the Riemann curvature tensor and \(R^\eD\) is the curvature of the superconnection \(\eD\). The Ricci tensor \(\Ric\in\odot^2T^*M\) is
\begin{equation}
	\Ric(X,Y)=\tr(Z\mapsto R^\nabla(X,Z)Y).
\end{equation}
In a local frame\footnote{Coordinate frame or local orthonormal frame. Recall that in the latter case, \(g_{\mu\nu}=\eta_{\mu\nu}\).} \(\{\be_\mu\}\), we define the components of the Riemann curvature tensor to be
\begin{equation}
	R^\nabla(\be_\mu,\be_\nu)\be_\tau = R\indices{_{\mu\nu}^\sigma_\tau}e_\sigma.
\end{equation}
We use the metric to raise and lower indices so that
\begin{equation}
	R_{\mu\nu\sigma\tau} = g_{\sigma\rho}R\indices{_{\mu\nu}^\rho_\tau} = g(\be_\sigma, R^\nabla(\be_\mu,\be_\nu)e_\tau).
\end{equation}
We denote the corresponding Riemann tensor in \(\wedge^2\odot^2 T^*M\) with these components by \(\Riem\):
\begin{equation}
	\Riem(W,X,Y,Z) = g(Y,R^\nabla(W,X)Z).
\end{equation}
The components of the Ricci tensor are then
\begin{equation}
	R_{\mu\nu} = \Ric(\be_\mu,\be_\nu) = R\indices{_{\mu\rho}^\rho_\nu},
\end{equation}
and the scalar curvature is the trace of the Ricci endomorphism defined by \(\Ric(X,Y)=g(X,\Ric(Y))\), or
\begin{equation}
	R = g^{\mu\nu}R_{\mu\nu}.
\end{equation}
Since the Riemann curvature tensor \(R^\nabla(X,Y)\) takes values in \(\fso(TM)\), it defines an element in the Clifford bundle as in \eqref{eq:so-spinor-action}, so the action on spinors fields \(s\) is given by
\begin{equation}
	R^\nabla(X,Y) \cdot s = \tfrac{1}{4}R_{\mu\nu\sigma\tau}X^\mu Y^\nu\Gamma^{\sigma\tau} s.
\end{equation} 
We will also use the Weyl tensor
\begin{equation}
	W = \Riem + \frac{R}{2n(n-1)}g\owedge g + \frac{1}{n-2}(\Ric-\tfrac{R}{n}g)\owedge g,
\end{equation}
which in components is
\begin{equation}
	W_{\mu\nu\rho\sigma} = R_{\mu\nu\rho\sigma} 
		+ \tfrac{1}{n-2}(R_{\mu\rho}g_{\nu\sigma} + R_{\nu\sigma}g_{\mu\rho} - R_{\mu\sigma}g_{\nu\rho} - R_{\nu\rho}g_{\mu\sigma}) 
		- \tfrac{R}{(n-1)(n-2)}(g_{\mu\rho}g_{\nu\sigma} - g_{\mu\sigma}g_{\nu\rho}).
\end{equation}

\subsection{Determination of the superconnection curvature}
\label{sec:curvature-calc}

From the definition of the curvature of the superconnection \(R^\eD\) and the Riemann curvature \(R^\nabla\), we have
\begin{equation}
  \begin{split}
    {R^\eD}(X,Y)s &= \eD_{\comm{X}{Y}}s - \comm{\eD_X}{\eD_Y}s \\
    &= \nabla_{\comm{X}{Y}}s - \beta_{\comm{X}{Y}}s -
    \comm{\nabla_X}{\nabla_Y}s - \comm{\beta_X}{\beta_Y}s +
    \comm{\nabla_X}{\beta_Y}s + \comm{\beta_X}{\nabla_Y}s \\
    &= R^\nabla(X,Y)\cdot s - \beta_{\comm{X}{Y}}s - \comm{\beta_X}{\beta_Y}s +
    \qty(\nabla_X \beta_Y)s - \qty(\nabla_Y \beta_X)s,
  \end{split}
\end{equation}
for all \(s\in\Gamma(\bbS)\). Locally, we can write
\begin{equation}\label{eq:curv-comps}
  {R^\eD}\indices{_{\mu\nu}^A_B} =
  \tfrac{1}{4}R\indices{_{\mu\nu\sigma\tau}}\Gamma\indices{^{\sigma\tau}}\delta^A_B
  - \comm{\beta_\mu}{\beta_\nu}\indices{^A_B} +
  2\nabla_{[\mu}\beta\indices{_{\nu]}^A_B}.
\end{equation}
We'll expand each of the \(\beta\) terms in turn. The components 
\(\beta\indices{_\mu^A_B}\) of \(\beta\) are defined by 
\((\beta_\mu s)^A = \beta\indices{_\mu^A_B}s^B\).
The differential term, straightforwardly, is
\begin{equation} \label{eq:curv-comps-diff-term}
  \begin{split}
    \nabla\indices{_{[\mu}}\beta\indices{_{\nu]}^A_B} & = -
    \tfrac{1}{8}\epsilon\indices{_{\alpha\beta\sigma\tau[\mu}}\nabla_{\nu]}C\indices{^{\alpha\beta}}
    \Gamma\indices{^{\sigma\tau}} \delta^A_B +
    \nabla_{[\mu}C\indices{_{\nu]\sigma}}\Gamma\indices{^\sigma}
    \delta^A_B - \tfrac{1}{4}\eta\indices{_{\sigma[\mu}}\nabla_{\nu]}
    F\indices{_\tau^A_B}\Gamma\indices{^{\sigma\tau}} -
    \tfrac{1}{2}\nabla_{[\mu} F\indices{_{\nu]}^A_B}.
  \end{split}
\end{equation}

For the commutator
\(\comm{\beta_\mu}{\beta_\nu}=2\qty({\beta_{[\mu}}{\beta_{\nu]}})\), we
first compute
\begin{equation}
  \begin{split}
    \beta\indices{_\mu^A_C}&\beta\indices{_\nu^C_B}\\
    & = \tfrac{1}{64}C\indices{^{\alpha\beta}}C\indices{^{\gamma\delta}}\qty(\epsilon\indices{_{\mu\alpha\beta\sigma\tau}}\Gamma\indices{^{\sigma\tau}} 
    + 8\eta\indices{_{\mu\alpha}}\Gamma\indices{_\beta}
    )\qty(
    \epsilon\indices{_{\nu\gamma\delta\kappa\lambda}}\Gamma\indices{^{\kappa\lambda}}
    + 8\eta\indices{_{\nu\gamma}} \Gamma_\delta
    )\delta^A_B
    \\
    &\qquad+ \tfrac{1}{16} F\indices{^\alpha^A_C} F\indices{^\beta^C_B} \qty(
    \Gamma\indices{_{\mu\alpha}} - 2 \eta_{\mu\alpha}
    )\qty(
    \Gamma\indices{_{\nu\beta}} - 2 \eta_{\nu\beta}
    )
    \\
    &\qquad + \tfrac{1}{32}C\indices{^{\alpha\beta}} F\indices{^\gamma^A_B}\qty[
    \qty(
    \epsilon\indices{_{\mu\alpha\beta\sigma\tau}}\Gamma\indices{^{\sigma\tau}}
    + 8\eta\indices{_{\mu\alpha}}\Gamma\indices{_\beta}
    ) 
    \qty( \Gamma\indices{_{\nu\gamma}} - 2 \eta_{\nu\gamma} )
    +\qty( \Gamma\indices{_{\mu\gamma}} - 2 \eta_{\mu\gamma} )
    \qty(
    \epsilon\indices{_{\nu\alpha\beta\sigma\tau}}\Gamma\indices{^{\sigma\tau}} 
    + 8\eta_{\nu\alpha}\Gamma\indices{_\beta}
    )
    ].
  \end{split}
\end{equation}

We will refer to the collections of terms in this expression
proportional to \(CC\), \(FF\) and \(CF\) as \([CC]\), \([ F F]\) and
\([C F]\) respectively. Expanding \([CC]\) and for now omitting the
\(\delta^A_B\), after a long calculation we have
\begin{equation}
  \begin{split}
    &\tfrac{1}{64}C\indices{^{\alpha\beta}}C\indices{^{\gamma\delta}} \qty[		\epsilon\indices{_{\mu\alpha\beta\sigma\tau}}\epsilon\indices{_{\nu\gamma\delta\kappa\lambda}}\Gamma\indices{^{\sigma\tau}}\Gamma\indices{^{\kappa\lambda}} 
    + 8\eta\indices{_{\nu\gamma}}\epsilon\indices{_{\mu\alpha\beta\sigma\tau}}\Gamma\indices{^{\sigma\tau}}\Gamma_\delta
    + 8\eta\indices{_{\mu\alpha}}\epsilon\indices{_{\nu\gamma\delta\kappa\lambda}}\Gamma\indices{_\beta}\Gamma\indices{^{\kappa\lambda}}
    + 64\eta\indices{_{\mu\alpha}}\eta\indices{_{\nu\gamma}}\Gamma\indices{_\beta}\Gamma_\delta
    ]
    \\
    &= \tfrac{1}{64}C\indices{^{\alpha\beta}}C\indices{^{\gamma\delta}}
    \epsilon\indices{_{\mu\alpha\beta\sigma\tau}}\epsilon\indices{_{\nu\gamma\delta\kappa\lambda}}
    \Gamma\indices{^{\sigma\tau\kappa\lambda}} 
    + \tfrac{1}{8}C\indices{^{\alpha\beta}}\qty(
    C\indices{_{\mu\rho}}\epsilon\indices{_{\nu\alpha\beta\sigma\tau}}
    + C\indices{_{\nu\rho}}\epsilon\indices{_{\mu\alpha\beta\sigma\tau}}
    )\Gamma\indices{^{\rho\sigma\tau}}
    \\
    & \qquad + \tfrac{5}{4}C\indices{_{\mu\alpha}}C\indices{_{\nu}^\alpha}
    - \tfrac{1}{8} \eta_{\mu\nu} C\indices{^{\alpha\beta}}C\indices{_{\alpha\beta}}
    \\
    & \qquad + \tfrac{1}{4}\qty(
    3C\indices{_{\mu\sigma}}C\indices{_{\nu\tau}}\Gamma\indices{^{\sigma\tau}} 
    - \qty(
    C\indices{_{\mu\alpha}}C\indices{_\sigma^\alpha}\Gamma\indices{_\nu^\sigma}
    - C\indices{_{\nu\alpha}}C\indices{_\sigma^{\alpha}}\Gamma\indices{_{\mu}^\sigma}
    )
    - \tfrac{1}{2}C\indices{^{\alpha\beta}}C\indices{_{\alpha\beta}}\Gamma\indices{_{\mu\nu}}
    )
    \\
    & \qquad 
    + \tfrac{1}{4}C\indices{^{\alpha\beta}}\qty(
    C\indices{_\mu^\gamma}\epsilon\indices{_{\nu\alpha\beta\gamma\tau}}
    - C\indices{_\nu^\gamma}\epsilon\indices{_{\mu\alpha\beta\gamma\tau}}
    )\Gamma^\tau.
  \end{split}
\end{equation}
                    
We have written the above expression so that terms symmetric in
\(\mu\nu\) appear first, followed by the
skew-symmetric terms. We will do the same with \([ F F]\) and \([C F]\).
For the former, after another calculation gives
\begin{equation}
  \begin{split}
    &\tfrac{1}{16} F\indices{^\alpha^A_C} F\indices{^\beta^C_B} \qty[
    \Gamma\indices{_{\mu\alpha}}\Gamma\indices{_{\nu\beta}} 
    - 2 \eta_{\nu\beta}\Gamma\indices{_{\mu\alpha}}
    - 2 \eta_{\mu\alpha}\Gamma\indices{_{\nu\beta}}
    + 4  \eta_{\mu\alpha} \eta_{\nu\beta}
    ]
    \\
    &=	- \tfrac{1}{32}\eta_{\mu\nu} \comm{ F_\alpha}{ F_\beta}\indices{^A_B}\Gamma\indices{^{\alpha\beta}}
    +\tfrac{5}{32}\acomm{ F_\mu}{ F_\nu}\indices{^A_B}
    -\tfrac{1}{16}\eta_{\mu\nu} \qty( F^\alpha F_\alpha)\indices{^A_B}
    \\
    & \qquad - \tfrac{1}{32}\comm{ F_\alpha}{ F_\beta}\indices{^A_B}\Gamma\indices{_{\mu\nu}^{\alpha\beta}} 
    -  \tfrac{1}{16}\qty( F^\alpha F_\alpha)\indices{^A_B}\Gamma\indices{_{\mu\nu}}
    + \tfrac{3}{32}\comm{ F_\mu}{ F_\nu}\indices{^A_B}
    \\
    & \qquad +\tfrac{1}{16}\qty[
    \qty( F_\nu F^\alpha - 2 F^\alpha F_\nu)\indices{^A_B}\Gamma\indices{_{\mu\alpha}}
    - \qty(2 F_\mu F^\alpha + F^\alpha F_\mu)\indices{^A_B}\Gamma\indices{_{\nu\alpha}}
    ],
  \end{split}
\end{equation}
where we have separated terms which are explicitly symmetric, skew-symmetric and of indeterminate symmetry in \(\mu\nu\). Turning finally to \([C F]\), we have
\begin{equation}
\begin{split}
  &\tfrac{1}{32}C\indices{^{\alpha\beta}} F\indices{^\gamma^A_B}\qty[
  \qty(
  \epsilon\indices{_{\mu\alpha\beta\sigma\tau}}\Gamma\indices{^{\sigma\tau}}
  + 8\eta\indices{_{\mu\alpha}}\Gamma\indices{_\beta} ) \qty(
  \Gamma\indices{_{\nu\gamma}} - 2 \eta\indices{_{\nu\gamma}} ) +\qty(
  \Gamma\indices{_{\mu\gamma}} - 2 \eta\indices{_{\mu\gamma}} ) \qty(
  \epsilon\indices{_{\nu\alpha\beta\sigma\tau}}\Gamma\indices{^{\sigma\tau}}
  + 8\eta\indices{_{\nu\alpha}}\Gamma\indices{_\beta} ) ]
  \\
	& = \tfrac{1}{16}\epsilon\indices{_{\alpha\beta\sigma\tau(\mu}}C\indices{^{\alpha\beta}} F\indices{^\gamma^A_B}
		\Gamma\indices{^{\sigma\tau}_{\nu)\gamma}}
	+ \tfrac{1}{2} F\indices{_\tau^A_B}C\indices{_{\sigma(\mu}}\Gamma\indices{_{\nu)}^{\sigma\tau}}
	- \tfrac{1}{4} C\indices{^{\alpha\beta}} F\indices{_{(\mu}^A_B} \epsilon\indices{_{\nu)\alpha\beta\sigma\tau}} 
		\Gamma\indices{^{\sigma\tau}}
	+ C{_{\sigma(\mu}} F\indices{_{\nu)}^A_B}\Gamma\indices{^\sigma}
\\
	&\qquad
		+ \tfrac{1}{8}C\indices{^{\alpha\beta}}\qty(
			-\epsilon\indices{_{\mu\nu\alpha\beta\sigma}} F\indices{_\tau^A_B}
			+ \eta\indices{_{\sigma[\mu}} \epsilon\indices{_{\nu]\alpha\beta\gamma\tau}} F\indices{^\gamma^A_B}
			)\Gamma\indices{^{\sigma\tau}}
		+ \tfrac{1}{2}\qty(
			C\indices{_{\mu\nu}} F\indices{_\sigma}
			+ \eta\indices{_{\sigma[\mu}}C\indices{_{\nu]\alpha}} F\indices{^\alpha^A_B}
			)\Gamma\indices{^{\sigma}}.
\end{split}
\end{equation}

We now have all of the terms of the commutator \(\comm{\beta_\mu}{\beta_\nu}\). This has two \(\Delta\)-indices, and after lowering both, it can be separated into symmetric and skew-symmetric parts with respect to these indices. The symmetric part is proportional to \(\epsilon\) and the skew-symmetric part takes values in \(\fsp(1)\). We must therefore determine the symmetry of the \(\Delta\) indices in each term. The following lemma addresses this.

\begin{lemma}
\(\comm{F_\mu}{F_\nu}_{AB}\) is symmetric in \(AB\) and \(\acomm{F_\mu}{F_\nu}_{AB}\) is skew-symmetric.
\end{lemma}

\begin{proof}
The result follows immediately from the fact that \(F\) takes values in \(\fsp(1)\). Alternatively, we can show this purely by an exercise in indices: note that \(F_{AB}\) is symmetric in \(AB\) and consider the product \(F_\mu F_\nu\) with lowered indices:
\begin{equation}
\begin{split}
	(F_\mu F_\nu)\indices{_A_B}
		&= F\indices{_\mu_A_C} F\indices{_\nu^C_B}
		= \epsilon^{DC}F\indices{_\mu_A_C} F\indices{_\nu_D_B}
		= -\epsilon^{CD}F\indices{_\nu_B_D} F\indices{_\mu_C_A}
		= -F\indices{_\nu_B_D} F\indices{_\mu^D_A}
		= -(F_\nu F_\mu)\indices{_B_A};
\end{split}
\end{equation}
we thus have
\begin{align}
	(F_{(\mu} F_{\nu)})_{AB} &= (F_{(\nu} F_{\mu)})_{AB} = -(F_{(\mu} F_{\nu)})_{BA},
	&
	(F_{[\mu} F_{\nu]})_{AB} &= -(F_{[\nu} F_{\mu]})_{AB} = (F_{[\mu} F_{\nu]})_{BA},
\end{align}
hence the result.
\end{proof}

We now evaluate the commutator. Since \(\comm{\beta_\mu}{\beta_\nu}= 2\qty(\beta_{[\mu}\beta_{\nu]})\), only the \(\mu\nu\)-skew-symmetric terms of \([CC]\), \([CF]\) and \([FF]\) contribute. Organising the terms of the commutator by their \(\Delta\)-symmetry but suppressing those indices for convenience, the commutator \(\comm{\beta_\mu}{\beta_\nu}\) is
\begin{equation}
\begin{split}
	\label{eq:curv-comps-comm-term}
	&\tfrac{1}{4}\Big[
			\qty(
				6C\indices{_{\mu\sigma}}C\indices{_{\nu\tau}}
			 	+ 4\eta\indices{_{\sigma[\mu}}C\indices{_{\nu]\alpha}}C\indices{_\tau^\alpha}
				-C\indices{^{\alpha\beta}}C\indices{_{\alpha\beta}}
					\eta\indices{_{\mu\sigma}}\eta\indices{_{\nu\tau}}
			)\epsilon
		+ \tfrac{1}{2}\qty(
				\eta\indices{_{\sigma[\mu}}\acomm{ F_{\nu]}}{ F_\tau}
				- \qty( F^\alpha F_\alpha)\eta\indices{_{\mu\sigma}}\eta\indices{_{\nu\tau}}
				)
		\Big]\Gamma\indices{^{\sigma\tau}} 
\\
	& \qquad 
	+ \epsilon\indices{_{\alpha\beta\gamma\sigma[\mu}} C\indices{^{\alpha\beta}}C\indices{^\gamma_{\nu]}}
		\epsilon\Gamma\indices{^\sigma}
	+\tfrac{1}{4}\qty[ 
		\eta\indices{_{\sigma[\mu}}\comm{ F_{\nu]}}{ F_\tau}
		- \qty(
			\epsilon\indices{_{\mu\nu\alpha\beta\sigma}}C\indices{^{\alpha\beta}} F\indices{_\tau}
			- \eta\indices{_{\sigma[\mu}} \epsilon\indices{_{\nu]\alpha\beta\gamma\tau}}C\indices{^{\alpha\beta}} F\indices{^\gamma}
			)
		]\Gamma\indices{^{\sigma\tau}}
\\
	&\qquad+ \qty[
		- \tfrac{1}{16}\epsilon\indices{_{\mu\nu\alpha\beta\sigma}} \comm{ F^\alpha}{ F^\beta}
		+ C\indices{_{\mu\nu}} F\indices{_\sigma}
		+ \eta\indices{_{\sigma[\mu}}C\indices{_{\nu]\alpha}} F\indices{^\alpha}
		]\Gamma\indices{^{\sigma}}
	+ \tfrac{3}{16}\comm{ F_\mu}{ F_\nu}.
\end{split}
\end{equation}

An explicit expression for the curvature of the superconnection in terms of \(C\) and \(F\) can now finally be found by substituting equations~\eqref{eq:curv-comps-diff-term} and \eqref{eq:curv-comps-comm-term} into equation~\eqref{eq:curv-comps}. For the sake of readability, we will first define some new notation. It will be useful to denote the components of \(R^\eD\) as follows:
\begin{equation}
	{R^\eD}\indices{_{\mu\nu}_{AB}} 
		= \tfrac{1}{2}L\indices{_{\mu\nu\sigma\tau}_{AB}}\Gamma\indices{^{\sigma\tau}}
		+ M\indices{_{\mu\nu\sigma}_{AB}}\Gamma\indices{^\sigma}
		+ N\indices{_{\mu\nu}_{AB}},
\end{equation}
where
\begin{equation}
\begin{aligned}
	L &= L^\wedge \epsilon_{AB} + L^\odot_{AB}
	&L^\wedge & \in\textstyle \Omega^2(M;\wedge^2 TM)
	& L^\odot & \in\textstyle \Omega^2(M;\wedge^2 TM\otimes\fsp(1))\\
	M &= M^\wedge \epsilon_{AB} + M^\odot_{AB}
	&M^\wedge & \in\textstyle \Omega^2(M;\wedge^1 TM)
	& M^\odot & \in\textstyle \Omega^2(M;\wedge^1 TM\otimes\fsp(1))\\
	N &= N^\wedge \epsilon_{AB} + N^\odot_{AB}
	&N^\wedge & \in\textstyle \Omega^2(M)
	&N^\odot & \in\textstyle \Omega^2(M; \fsp(1)).
\end{aligned}
\end{equation}
Said another way, \(L\indices{_{\mu\nu\sigma\tau}_{AB}},M\indices{_{\mu\nu\sigma}_{AB}}\) and \(N\indices{_{\mu\nu}_{AB}}\) are skew-symmetric in \(\mu\nu\), and \(L\indices{_{\mu\nu\sigma\tau}_{AB}}\) is also skew-symmetric in \(\sigma\tau\); \(L^\wedge\) and \(L^\odot\) are respectively the skew-symmetric and symmetric parts of \(L\) (with respect to \(\Delta\)-indices).

Since it is skew-symmetric, the anticommutator \(\acomm{ F_\mu}{ F_\nu}\) can be written in terms of its trace: 
	\(\acomm{ F_\mu}{ F_\nu}\indices{_{AB}} = \tfrac{1}{2}\acomm{ F_\mu}{ F_\nu}\indices{^C_C}\epsilon{_{AB}}\).
Let us define 
	\(F_\mu \cdot F_\nu = F\indices{_{\mu}^{AB}}F\indices{_\nu_{AB}}= -( F_\mu F_\nu)\indices{^A_A}\) 
and \(F^2 = F^\alpha \cdot F_\alpha\). Note that \(F_\mu \cdot F_\nu = F_\nu \cdot F_\mu\), so
\begin{align}
	\acomm{ F_\mu}{ F_\nu}\indices{_{AB}} &= -\qty( F_\mu\cdot F_\nu)\epsilon{_{AB}}
	&& \text{and}
	& \qty( F_\alpha F^\alpha)\indices{_{AB}} &= -\tfrac{1}{2} F^2\epsilon{_{AB}}.
\end{align}

We thus have
\begin{align}
	\begin{split}
	{L^\wedge}\indices{_{\mu\nu\sigma\tau}} & = \tfrac{1}{2}\Big[
		R\indices{_{\mu\nu\sigma\tau}}
		- \qty(
			6 C\indices{_{\mu[\sigma}}C\indices{_{|\nu|\tau]}}
			+ 2 \eta\indices{_{\mu[\sigma}}C\indices{_{|\nu\alpha|}}C\indices{_{\tau]}^\alpha}
			- 2 \eta\indices{_{\nu[\sigma}}C\indices{_{|\mu\alpha|}}C\indices{_{\tau]}^\alpha}
			- C\indices{^{\alpha\beta}}C\indices{_{\alpha\beta}}
				\eta\indices{_{\mu[\sigma}}\eta\indices{_{|\nu|\tau]}}
				)
\\
	& \qquad\qquad \qquad \qquad
		- \epsilon\indices{_{\alpha\beta\sigma\tau[\mu}}\nabla_{\nu]}C\indices{^{\alpha\beta}}
		+ \tfrac{1}{4}\qty(
			\eta\indices{_{\mu[\sigma}}(F_{|\nu|}\cdot F_{\tau]})
			- \eta\indices{_{\nu[\sigma}}(F_{|\mu|}\cdot F_{\tau]})
			- F^2 \eta\indices{_{\mu[\sigma}}\eta\indices{_{|\nu|\tau]}}
			)
		\Big]
	\end{split}
\label{eq:Lwedge}
\\
	{M^\wedge}\indices{_{\mu\nu\sigma}} & =
		-\epsilon\indices{_{\alpha\beta\gamma\sigma[\mu}} 
			C\indices{^{\alpha\beta}}C\indices{^\gamma_{\nu]}}
		+ 2\nabla_{[\mu}C\indices{_{\nu]\sigma}}
\label{eq:Mwedge}
\\
	{N^\wedge}\indices{_{\mu\nu}} &= 0
\label{eq:Nwedge}
\\
	{L^\odot}\indices{_{\mu\nu\sigma\tau}} & =
		-\tfrac{1}{4}\qty(
		\eta\indices{_{\mu[\sigma}}\comm{ F_{|\nu|}}{ F_{\tau]}}
		-\eta\indices{_{\nu[\sigma}}\comm{ F_{|\mu|}}{ F_{\tau]}}
		)
		+ \tfrac{1}{2}
			\epsilon\indices{_{\mu\nu\alpha\beta[\sigma}}
				C\indices{^{\alpha\beta}} F\indices{_{\tau]}}
\\
		& \qquad - \tfrac{1}{4}\qty(
			\eta\indices{_{\mu[\sigma}}\epsilon\indices{_{\tau]\nu\alpha\beta\gamma}}
				C\indices{^{\alpha\beta}} F\indices{^\gamma}
			- \eta\indices{_{\nu[\sigma}}\epsilon\indices{_{\tau]\mu\alpha\beta\gamma}}
			C\indices{^{\alpha\beta}} F\indices{^\gamma}
			)
		- \tfrac{1}{2}\qty(
			\eta\indices{_{\mu[\sigma}}\nabla_{|\nu|} F\indices{_{\tau]}}
			- \eta\indices{_{\nu[\sigma}}\nabla_{|\mu|} F\indices{_{\tau]}}
			)
\label{eq:Lodot}
\\
	{M^\odot}\indices{_{\mu\nu\sigma}} & =
		\tfrac{1}{16}\epsilon\indices{_{\mu\nu\alpha\beta\sigma}}\comm{ F^\alpha}{ F^\beta}
		- C\indices{_{\mu\nu}} F\indices{_\sigma}
		- \eta\indices{_{\sigma[\mu}}C\indices{_{\nu]\alpha}} F\indices{^\alpha}
\label{eq:Modot}
\\
	{N^\odot}\indices{_{\mu\nu}} & =
		- \tfrac{3}{16}\comm{ F_\mu}{ F_\nu} - \nabla_{[\mu} F\indices{_{\nu]}}.
\label{eq:Nodot}
\end{align}

This explicitly gives the superconnection curvature \(R^\eD\) in terms of the background fields. We now go on to show how equations of motion and maximal supersymmetry conditions can be extracted from the curvature.

\subsection{Clifford trace of superconnection curvature}
\label{sec:gamma-trace}

In this section we seek to compute necessary and sufficient conditions
for the Clifford trace of the curvature
\(\Gamma^\nu{R^\eD}\indices{_{\mu\nu}}\) to vanish identically as a
one-form with values in $\End(\bbS)$.   To this end we will make use of
some of the identities from Appendix~\ref{sec:tensor-id}.

\begin{theorem}\label{thm:gamma-trace-zero}
Let \((M,g,C,F)\) be a 5-dimensional background with superconnection \(\eD\) given by equation~\eqref{eq:superconnection}. The Clifford trace of the curvature \(\Gamma^\nu{R^\eD}\indices{_{\mu\nu}}\) vanishes if and only if the following equations hold:
\begin{gather}
	\nabla\indices{^{\alpha}}C\indices{_{\alpha\mu}}
		=\tfrac{1}{2}\epsilon\indices{_{\mu\alpha\beta\gamma\delta}}
			C\indices{^{\alpha\beta}}C\indices{^{\gamma\delta}}
\label{eq:sugra-maxwell}
\\
	C\indices{_{\mu\alpha}} F\indices{^\alpha}  = 0
\label{eq:C-F-contraction}
\\
	\nabla\indices{_{[\sigma}}C\indices{_{\mu\nu]}}  = 0
\label{eq:C-closed-components}
\\
	R\indices{_{\mu\nu}}
	+ \qty(
		6C\indices{_{\mu\alpha}}C\indices{_\nu^\alpha}
		- \eta\indices{_{\mu\nu}}C\indices{^{\alpha\beta}}C\indices{_{\alpha\beta}}
		)
	- \tfrac{3}{8}\qty(
		\qty( F_\mu\cdot F_\nu) 
		- \eta\indices{_{\mu\nu}} F^2
		)
	= 0
\label{eq:sugra-einstein}
\\
	\nabla\indices{_{\mu}} F\indices{_{\nu}} 
		= -\tfrac{1}{2}\epsilon\indices{_{\mu\nu\alpha\beta\gamma}}
			C\indices{^{\alpha\beta}} F\indices{^\gamma}
\label{eq:C-F-eom}
\\
	\comm{ F_\mu}{ F_\nu}  = 0.
\label{eq:F-comm-zero}
\end{gather}
\end{theorem}

Equation~\eqref{eq:C-closed-components} is simply \(dC=0\), or \(C\) is
closed, consistent with \(C\) being the field strength tensor of a
one-form potential as in
supergravity. Equation~\eqref{eq:C-F-contraction} is simply
\(\iota_{F\indices{_{AB}}}C=0\). Equation~\eqref{eq:sugra-maxwell} is
the Maxwell-like supergravity bosonic equation of motion, while if we
set \( F=0\), equation~\eqref{eq:sugra-einstein} becomes the
Einstein-like equation. Equation~\eqref{eq:C-F-eom} provides a third
equation of motion wherever \( F\neq 0\).  Writing \(\Delta\)-indices,
we also learn that each \(F\indices{_{AB}}\) is a Killing vector field
and (from equation~\eqref{eq:F-comm-zero}) that the \( F_\mu\) commute
everywhere under the \(\fsp(1)\) commutator.

\begin{proof}
Using \(L,M,N\) as defined above,
\begin{equation}\label{eq:leftgtrace}
\begin{split}
	\Gamma\indices{^\nu}{R^\eD}\indices{_{\mu\nu}}
		&= \tfrac{1}{2}L\indices{_{\mu\nu\sigma\tau}}\qty
			(\Gamma\indices{^{\nu\sigma\tau}}+2\eta\indices{^{\nu[\sigma}}\Gamma\indices{^{\tau]}})
		+ M\indices{_{\mu\nu\sigma}}\qty(\Gamma\indices{^{\nu\sigma}}+\eta\indices{^{\nu\sigma}})
		+ N\indices{_{\mu\nu}}\Gamma\indices{^\nu},
\\
		&= \qty(
			-\tfrac{1}{4}\epsilon\indices{_{\alpha\beta\gamma\sigma\tau}}
				L\indices{_\mu^{\alpha\beta\gamma}}
			+ M\indices{_{\mu\sigma\tau}}
			)\Gamma\indices{^{\sigma\tau}}
		+ \qty(
			L\indices{_\mu^\alpha_\alpha_\sigma} 
			+ N\indices{_{\mu\sigma}}
			)\Gamma\indices{^\sigma}
		+ M\indices{_\mu^\alpha_\alpha}		
\end{split}
\end{equation}

For the first term we have
\begin{align}
	{L^\wedge}\indices{_{\mu[\alpha\beta\gamma]}}
	& = 3C\indices{_{\mu[\alpha}}C\indices{_{\beta\gamma]}}
	+ \tfrac{1}{4}\qty(
		\epsilon\indices{_{\kappa\lambda\alpha\beta\gamma}}\nabla_{\mu}C\indices{^{\kappa\lambda}}
		- \epsilon\indices{_{\mu\kappa\lambda[\alpha\beta}}\nabla_{\gamma]}C\indices{^{\kappa\lambda}}
		),
\\
	{L^\odot}\indices{_{\mu[\alpha\beta\gamma]}} & =
		\tfrac{1}{4}\eta\indices{_{\mu[\alpha}}\comm{ F_{\beta}}{ F_{\gamma]}}
		+ \tfrac{1}{2}\epsilon\indices{_{\mu\kappa\lambda[\alpha\beta}}
			C\indices{^{\kappa\lambda}} F\indices{_{\gamma]}}
		- \tfrac{1}{4}\eta\indices{_{\mu[\alpha}}\epsilon\indices{_{\beta\gamma]\kappa\lambda\rho}}
			C\indices{^{\kappa\lambda}} F\indices{^\rho}
		+ \tfrac{1}{2}\eta\indices{_{\mu[\alpha}}\nabla_{\beta} F\indices{_{\gamma]}},
\end{align}
thus
\begin{align}
	-\tfrac{1}{4}\epsilon\indices{_{\alpha\beta\gamma\sigma\tau}}
		{L^\wedge}\indices{_\mu^{\alpha\beta\gamma}} 
	&= -\tfrac{3}{4}\epsilon\indices{_{\alpha\beta\gamma\sigma\tau}}
			C\indices{_\mu^\alpha}C\indices{^{\beta\gamma}}
		-\tfrac{1}{2}\nabla_\mu C\indices{_{\sigma\tau}} 
		-\tfrac{1}{2}\eta\indices{_{\mu[\sigma}}\nabla^\gamma C\indices{_{|\gamma|\tau]}},
\\
	-\tfrac{1}{4}\epsilon\indices{_{\alpha\beta\gamma\sigma\tau}}
	{L^\odot}\indices{_\mu^{\alpha\beta\gamma}} 
	&= - \tfrac{1}{16}\epsilon\indices{_{\mu\alpha\beta\sigma\tau}}\comm{ F^\alpha}{ F^\beta}
		- \tfrac{1}{4} C\indices{_{\sigma\tau}} F\indices{_\mu}
		+ \tfrac{1}{2} C\indices{_{\mu[\sigma}} F\indices{_{\tau]}}
		- \eta_{\mu[\sigma} C\indices{_{\tau]\gamma}} F\indices{^\gamma}
		-\tfrac{1}{8} \epsilon\indices{_{\mu\alpha\beta\sigma\tau}}
			\nabla\indices{^\alpha} F\indices{^\beta}.
\end{align}

The trace in the second term is
\begin{align}
	{L^\wedge}\indices{_\mu^\alpha_\alpha_\sigma} 
	& = \tfrac{1}{2}\qty[
		R\indices{_{\mu\sigma}}
		+ \qty(
			6C\indices{_{\mu\alpha}}C\indices{_\sigma^\alpha}
			- \eta\indices{_{\mu\sigma}}C\indices{^{\alpha\beta}}C\indices{_{\alpha\beta}}
			)
		- \tfrac{3}{8}\qty(
			(F_\mu\cdot F_\sigma)
			- \eta\indices{_{\mu\sigma}} F^2
			)
		+ \tfrac{1}{2}\epsilon\indices{_{\mu\sigma\alpha\beta\gamma}}\nabla^\alpha
			C\indices{^{\beta\gamma}}
		],
\\
	{L^\odot}\indices{_\mu^\alpha_\alpha_\sigma} 
	& =
	\tfrac{3}{8}\comm{ F_{\mu}}{ F_\sigma}
	- \tfrac{1}{8}\epsilon\indices{_{\mu\sigma\alpha\beta\gamma}}
		C\indices{^{\alpha\beta}} F\indices{^\gamma}
	+ \tfrac{1}{4}\qty(
		\eta\indices{_{\mu\sigma}}\nabla_{\alpha} F\indices{^\alpha} + 3 \nabla_\mu F_\sigma
		).
\end{align}

We also need
\begin{align}
	{M^\wedge}\indices{_{\mu[\sigma\tau]}}
	& = -\tfrac{1}{2}\epsilon\indices{_{\mu\alpha\beta\gamma[\sigma}} 
		C\indices{_{\tau]}^\alpha}C\indices{^{\beta\gamma}}
	+ \tfrac{1}{2}\epsilon\indices{_{\alpha\beta\gamma\sigma\tau}} 
		C\indices{^\alpha_\mu}C\indices{^{\beta\gamma}}
	+ \nabla\indices{_\mu}C\indices{_{\sigma\tau}}
	+ \nabla\indices{_{[\sigma}}C\indices{_{\tau]\mu}},
\\
	{M^\wedge}\indices{_{\mu}^{\alpha}_{\alpha}}
	& = -\tfrac{1}{2}\epsilon\indices{_{\mu\alpha\beta\gamma\delta}}
		C\indices{^{\alpha\beta}}C\indices{^{\gamma\delta}}
	+ \nabla\indices{^{\alpha}}C\indices{_{\alpha\mu}},
\\
	{M^\odot}\indices{_{\mu[\sigma\tau]}} 
	& = \tfrac{1}{16}\epsilon\indices{_{\mu\sigma\tau\alpha\beta}} \comm{ F^\alpha}{ F^\beta}
		- C\indices{_{\mu[\sigma}} F\indices{_{\tau]}}
		+ \tfrac{1}{2}\eta\indices{_{\mu[\sigma}}C\indices{_{\tau]\alpha}} F\indices{^\alpha},
\\
	{M^\odot}\indices{_{\mu}^{\alpha}_{\alpha}}
		&= C\indices{_{\mu\alpha}} F\indices{^\alpha}.
\end{align}

It follows immediately from equation~\eqref{eq:leftgtrace} that \(\Gamma\indices{^\nu}{R^\eD}\indices{_{\mu\nu}}\)  vanishes if and only if
\begin{equation}
\begin{aligned}
	-\tfrac{1}{4}\epsilon\indices{_{\alpha\beta\gamma\sigma\tau}}
		{L^\wedge}\indices{_\mu^{\alpha\beta\gamma}}
	+ {M^\wedge}\indices{_{\mu[\sigma\tau]}}
	&= 0,
&
	-\tfrac{1}{4}\epsilon\indices{_{\alpha\beta\gamma\sigma\tau}}
		{L^\odot}\indices{_\mu^{\alpha\beta\gamma}}
	+ {M^\odot}\indices{_{\mu[\sigma\tau]}}
	&= 0,
\\	
	{L^\wedge}\indices{_\mu^\alpha_\alpha_\sigma}
	&= 0,
&
	{L^\odot}\indices{_\mu^\alpha_\alpha_\sigma} 
	+ {N^\odot}\indices{_{\mu\sigma}}
	&= 0,
\\	
	{M^\wedge}\indices{_\mu^\alpha_\alpha}
	&=0,
&
	{M^\odot}\indices{_\mu^\alpha_\alpha}
	&=0.
\end{aligned}
\end{equation}

Expanding each of these out, we obtain equations~\eqref{eq:sugra-maxwell} and \eqref{eq:C-F-contraction} as well as
\begin{gather}
	-\tfrac{1}{4}\epsilon\indices{_{\alpha\beta\gamma\sigma\tau}}
		C\indices{_\mu^\alpha}C\indices{^{\beta\gamma}}
	-\tfrac{1}{2}\epsilon\indices{_{\mu\alpha\beta\gamma[\sigma}} 
		C\indices{_{\tau]}^\alpha}C\indices{^{\beta\gamma}}
	+\tfrac{3}{2}\nabla\indices{_{[\mu}} C\indices{_{\sigma\tau]}} 
	-\tfrac{1}{2}\eta\indices{_{\mu[\sigma}}\nabla^\gamma C\indices{_{|\gamma|\tau]}}
	 = 0,
	\label{eq:gamma-trace-w2w2}
\\
	R\indices{_{\mu\sigma}}
	+ \qty(
		6C\indices{_{\mu\alpha}}C\indices{_\sigma^\alpha}
		- \eta\indices{_{\mu\sigma}}C\indices{^{\alpha\beta}}C\indices{_{\alpha\beta}}
		)
	- \tfrac{3}{8}\qty((F_\mu\cdot F_\sigma)- \eta\indices{_{\mu\sigma}} F^2)
	+ \tfrac{1}{2}\epsilon\indices{_{\mu\sigma\alpha\beta\gamma}}\nabla^\alpha
		C\indices{^{\beta\gamma}}
	 = 0,
	\label{eq:gamma-trace-w1w2}
\\
	- \tfrac{3}{4} C\indices{_{[\sigma\tau}} F\indices{_{\mu]}}
	- \tfrac{1}{2}\eta_{\mu[\sigma} C\indices{_{\tau]\gamma}} F\indices{^\gamma}
	-\tfrac{1}{8} \epsilon\indices{_{\mu\alpha\beta\sigma\tau}}
		\nabla\indices{^\alpha} F\indices{^\beta}
	 = 0,
	 \label{eq:gamma-trace-w2s2}
\\
	\tfrac{3}{16}\comm{ F_{\mu}}{ F_\sigma}
	-\tfrac{1}{4}\nabla\indices{_{[\mu}} F\indices{_{\sigma]}}
	- \tfrac{1}{8}\epsilon\indices{_{\mu\sigma\alpha\beta\gamma}}
		C\indices{^{\alpha\beta}} F\indices{^\gamma}
	+ \tfrac{1}{4}\qty(
		\eta\indices{_{\mu\sigma}}\nabla_{\alpha} F\indices{^\alpha} + 3 \nabla\indices{_{(\mu}} F\indices{_{\sigma)}}
		)
	 = 0.
	 \label{eq:gamma-trace-w1s2}
\end{gather}

We'll tackle these equations by splitting each of them into irreducible \(\mathfrak{so}(V)\)-module components. Taking equation~\eqref{eq:gamma-trace-w2w2} first, its full skew-symmetrisation is
\begin{equation}
	\tfrac{1}{4}\epsilon\indices{_{\alpha\beta\gamma[\sigma\tau}}
		C\indices{_{\mu]}^\alpha}C\indices{^{\beta\gamma}}
	+\tfrac{3}{2}\nabla\indices{_{[\mu}} C\indices{_{\sigma\tau]}}
	 = 0.
\end{equation}
The first term vanishes identically by \eqref{eq:tensor-id-AA-1} from Appendix~\ref{sec:tensor-id}, so this is equivalent to equation~\eqref{eq:C-closed-components}. The trace of \eqref{eq:gamma-trace-w2w2} is
\begin{equation}
	\tfrac{1}{2}\epsilon\indices{_{\tau\alpha\beta\gamma\delta}}
		C\indices{^{\alpha\beta}}C\indices{^{\gamma\delta}}
	-\nabla^\alpha C\indices{_{\alpha\tau}}
	 = 0,
\end{equation}
which gives us equation~\eqref{eq:sugra-maxwell} again. What remains of \eqref{eq:gamma-trace-w2w2} after subtracting off the skew and trace parts is
\begin{equation}
	-\tfrac{1}{4}\epsilon\indices{_{\alpha\beta\gamma\sigma\tau}}
		C\indices{_\mu^\alpha}C\indices{^{\beta\gamma}}
	-\tfrac{1}{2}\epsilon\indices{_{\mu\alpha\beta\gamma[\sigma}} 
		C\indices{_{\tau]}^\alpha}C\indices{^{\beta\gamma}}
	-\tfrac{1}{4}\eta\indices{_{\mu[\sigma}}\epsilon\indices{_{\tau]\alpha\beta\gamma\delta}}
		C\indices{^{\alpha\beta}}C\indices{^{\gamma\delta}}
	 = 0,
\end{equation}
which using the tensor identities \eqref{eq:tensor-id-AA-1} and \eqref{eq:tensor-id-AA-2} from Appendix~\ref{sec:tensor-id} is trivial. The skew-symmetric part of equation~\eqref{eq:gamma-trace-w1w2} gives equation~\eqref{eq:C-closed-components} again. The symmetric part is equation~\eqref{eq:sugra-einstein}. The totally skew-symmetric part of equation~\eqref{eq:gamma-trace-w2s2} is
\begin{equation}
	- \tfrac{3}{4} C\indices{_{[\sigma\tau}} F\indices{_{\mu]}}
	-\tfrac{1}{8} \epsilon\indices{_{\mu\sigma\tau\alpha\beta}}
		\nabla\indices{^\alpha} F\indices{^\beta}
	 = 0.
	\label{eq:gamma-trace-w2s2w3}
\end{equation}
The trace gives equation~\eqref{eq:C-F-contraction} again. The remaining part vanishes identically. The skew-symmetric part of equation~\eqref{eq:gamma-trace-w1s2} is
\begin{equation}
	\tfrac{3}{16}\comm{ F_{\mu}}{ F_\sigma}
	-\tfrac{1}{4}\nabla\indices{_{[\mu}} F\indices{_{\sigma]}}
	- \tfrac{1}{8}\epsilon\indices{_{\mu\sigma\alpha\beta\gamma}}
		C\indices{^{\alpha\beta}} F\indices{^\gamma}
	 = 0,
	\label{eq:gamma-trace-w1s2w2}
\end{equation}
and the symmetric part is
\begin{equation}
	\tfrac{1}{4}\qty(
		\eta\indices{_{\mu\sigma}}\nabla_{\alpha} F\indices{^\alpha} 
		+ 3 \nabla\indices{_{(\mu}} F\indices{_{\sigma)}}
		)
	 = 0.
\end{equation}
Tracing this equation gives \(\nabla_\alpha F^\alpha=0\), which can be
back-substituted to give \(\nabla_{(\mu}F_{\sigma)}=0\). Thus
contracting equation~\eqref{eq:gamma-trace-w2s2w3} with
\(\epsilon^{\kappa\lambda\mu\sigma\tau}\) yields
equation~\eqref{eq:C-F-eom}. Substituting that equation into
equation~\eqref{eq:gamma-trace-w1s2w2} finally produces
equation~\eqref{eq:F-comm-zero}.
\end{proof}

\subsection{Flatness of the superconnection and maximal supersymmetry}
\label{sec:flat-connection}

We now turn to finding conditions on \((M,g,C,F)\) equivalent to the
vanishing of the curvature of the superconnection, which by
Proposition~\ref{prop:max-susy-vanishing-curvature} is necessary for the
background to be maximally supersymmetric. We will prove the following
theorem.

\begin{theorem}\label{thm:maximalsusy}
  Let \((M,g,C,F)\) be a 5-dimensional background with superconnection
  \(\eD\) given by Equation~\eqref{eq:superconnection}. If it is
  maximally supersymmetric then \(C=0\) or \(F=0\), and
  \begin{enumerate}
  \item If \(C=0\) and \(F=0\) then the Riemann curvature vanishes.
  \item If \(C = 0\) and \(F\neq0\) then \(F = \varphi\otimes r\) for some parallel one-form \(\varphi\) and some \(r\in\fsp(1)\), and the Riemann curvature is given by
    \begin{equation}
      \label{eq:max-susy-riemann-F}
      R_{\mu\nu\sigma\tau}
      = \varphi^2 \eta\indices{_{\mu[\sigma}}\eta\indices{_{|\nu|\tau]}}
      - \qty(\eta\indices{_{\mu[\sigma}}\varphi_{|\nu|}\varphi_{\tau]}
      - \eta\indices{_{\nu[\sigma}}\varphi_{|\mu|}\varphi_{\tau]}).
    \end{equation}
  \item If \(F=0\) and \(C\neq0\) then \(C\) is a closed 2-form such that
    \begin{equation}
      \label{eq:strong-sugra-maxwell-trace}
      \nabla_\sigma C_{\mu\nu} 
      = \tfrac{1}{2}\eta\indices{_{\sigma[\mu}}\nabla\indices{^\alpha} C\indices{_{|\alpha|\nu]}}
      =\tfrac{1}{4}\eta\indices{_{\sigma[\mu}}\epsilon\indices{_{\nu]\alpha\beta\gamma\delta}}
      C^{\alpha\beta}C^{\gamma\delta},
    \end{equation}
    and the Riemann curvature is given by
    \begin{equation}
      \label{eq:max-susy-riemann-C}
      R_{\mu\nu\sigma\tau} 
      = 2 C\indices{_{\mu\nu}}C\indices{_{\sigma\tau}}
      + 2 C\indices{_{\mu[\sigma}}C\indices{_{|\nu|\tau]}}
      + 2 \eta\indices{_{\mu[\sigma}}C\indices{_{|\nu\alpha|}}C\indices{_{\tau]}^\alpha}
      - 2 \eta\indices{_{\nu[\sigma}}C\indices{_{|\mu\alpha|}}C\indices{_{\tau]}^\alpha}
      - C\indices{^{\alpha\beta}}C\indices{_{\alpha\beta}}
      \eta\indices{_{\mu[\sigma}}\eta\indices{_{|\nu|\tau]}}.
    \end{equation}
  \end{enumerate}
\end{theorem}

\begin{proof}
The curvature \({R^\eD}\indices{_{\mu\nu}}\) vanishes if and only if each of the components, given by equations~\eqref{eq:Lwedge}-\eqref{eq:Nodot}, vanishes separately. We first simplify the problem by noting that if \({R^\eD}\) vanishes, in particular its Clifford trace vanishes, so we may use the conditions derived at the end of the previous subsection. We will also need some of the identities of Appendix~\ref{sec:tensor-id}. Indeed, using \(\comm{ F_\mu}{ F_\nu}=0\) and \(C_{\mu\alpha} F^\alpha=0\), the vanishing of \(M^\odot\) and \(N^\odot\) gives
\begin{align}
	\nabla\indices{_{[\mu}} F\indices{_{\nu]}} &= 0	&&\text{and} 	& C_{\mu\nu} F_\sigma &= 0,
\end{align}
so in particular, either \( F=0\) or \(C=0\) and since we already have \(\nabla_{(\mu}F_{\nu)}=0\) from the zero Clifford trace equations, \(F\) is parallel. Now \(L^\odot=0\) exactly. Using \(dC=0\) and equation~\eqref{eq:tensor-id-AA-1}, the vanishing of \(M^{\wedge}\) is equivalent to 
\begin{equation}
	\label{eq:strong-sugra-maxwell}
	\nabla_\sigma C_{\mu\nu} 
		= \tfrac{1}{2}\epsilon\indices{_{\alpha\beta\gamma\mu\nu}}
			C\indices{^{\alpha\beta}}C\indices{^\gamma_\sigma},
\end{equation}
and then using equations~\eqref{eq:tensor-id-AA-2} and \eqref{eq:sugra-maxwell} (which is in fact the trace of equation~\eqref{eq:strong-sugra-maxwell}) yields equation~\eqref{eq:strong-sugra-maxwell-trace}. Now using equation~\eqref{eq:strong-sugra-maxwell-trace}, \(L^\wedge=0\) if and only if
\begin{equation}
\begin{split}
	R\indices{_{\mu\nu\sigma\tau}} 
			&= 2 C\indices{_{\mu\nu}}C\indices{_{\sigma\tau}}
			+ 2 C\indices{_{\mu[\sigma}}C\indices{_{|\nu|\tau]}}
			+ 2 \eta\indices{_{\mu[\sigma}}C\indices{_{|\nu\alpha|}}C\indices{_{\tau]}^\alpha}
			- 2 \eta\indices{_{\nu[\sigma}}C\indices{_{|\mu\alpha|}}C\indices{_{\tau]}^\alpha}
			- C\indices{^{\alpha\beta}}C\indices{_{\alpha\beta}}
				\eta\indices{_{\mu[\sigma}}\eta\indices{_{|\nu|\tau]}}
\\
		& \qquad - \tfrac{1}{4}\qty(
			\eta\indices{_{\mu[\sigma}}\qty( F_{|\nu|}\cdot F_{\tau]})
			- \eta\indices{_{\nu[\sigma}}( F_{|\mu|}\cdot F_{\tau]})
			- F^2 \eta\indices{_{\mu[\sigma}}\eta\indices{_{|\nu|\tau]}}
			),
\end{split}
\end{equation}
which we note can be traced to give equation~\eqref{eq:sugra-einstein}. For the \(F=0\) case, we get equation~\eqref{eq:max-susy-riemann-C}. Finally, for the \(C=0\) case, since both \(\nabla F=0\) and \(\comm{F_\mu}{F_\nu}=0\), \(F=\varphi\otimes r\), for some one-form \(\varphi\) and some fixed \(r\in\fsp(1)\). In components, \(F_\mu^{AB}=\varphi_{\mu} r^{AB}\), so \(F_\mu\cdot F_\nu = \varphi_\mu\varphi_\nu r^{AB}r_{AB}\). We have \(r^{AB}r_{AB}=-\tr(rr)\), where this trace is taken in the vector representation of \(\fsp(1)\cong \fsu(2)\) on \(\CC^2\), and is therefore negative-definite. We can assume without loss of generality that \(r\neq 0\) (by choosing \(\varphi\)=0 if \(F=0\)), and then by rescaling \(r\) and \(\varphi\) we can also assume that \(r^{AB}r_{AB}=-\tr(rr)=4\), yielding equation~\eqref{eq:max-susy-riemann-F}.
\end{proof}

\section{The Killing superalgebra}
\label{sec:killing-superalg}

In this section we define and prove the existence of the supersymmetry algebra of a supersymmetric background. To do so, we will need the notion of the spinorial derivative, for which we will give a definition and state some properties. We will also need to upgrade our definitions of spinor bilinears from Section~\ref{sec:spin-conv} to bilinears of spinor fields and derive some of their differential properties.

\subsection{The spinorial Lie derivative}
\label{sec:spin-lie-der}

Throughout, let \((M,g)\) be a spin manifold with an associated spinor bundle \(\bbS\).

Recall that for any vector field \(X\), \(\nabla X\) defines an endomorphism on the \(C^\infty(M)\)-module of vector fields \(\fX(M)\) by \(Y \mapsto \nabla_Y X\); \(\nabla X\) is a section of \(\End(TM)\) with components \((\nabla X)\indices{^\mu_\nu} = \nabla_\nu X^\mu\). Furthermore, \(\nabla X\) is actually a section of  \(\mathfrak{so}(TM)\) if and only if \(X\) is a Killing vector field (this follows directly from the definition), in which case \(\nabla X\) has the action on spinors \((\nabla X) s = -\tfrac{1}{4}\nabla_\mu X_\nu \Gamma^{\mu\nu}s\). This allows us to make the following definition.

\begin{definition}[\cite{MR0312413,FigueroaO'Farrill:1999va}]
The \textbf{spinorial Lie derivative} of a spinor field \(s\) along the Killing vector field \(X\) is  given by
\begin{equation}
	\eL_X s = \nabla_X s - (\nabla X) s.
\end{equation}
\end{definition}
Locally, we have
\begin{equation}
	\eL_X s = \nabla_X s + \tfrac{1}{4} \nabla_\mu X_\nu \Gamma^{\mu\nu} s.
\end{equation}

This obeys the Leibniz rule: for a smooth function \(f\) and a spinor field \(s\),
\begin{equation}
\begin{split}
	\eL_X (f s) 
		&= \nabla_X(fs) - (\nabla X) (fs) \\
		&= X(f)s + f\nabla_X s - f((\nabla X) s)\\
		&= X(f)s + f\eL_X s.
\end{split}
\end{equation}

We will not prove the following lemma; we refer the reader to \cite{MR0312413,FigueroaO'Farrill:1999va}.

\begin{lemma}\label{prop:liederprops}
The spinorial Lie derivative obeys the following properties:
\begin{enumerate}
	\item Representation of the Lie algebra of vector fields on spinor fields:
	\begin{equation}
		\eL_X\eL_Ys = \eL_{\eL_X Y} s + \eL_Y \eL_X s
	\end{equation}
	for all Killing vectors \(X,Y\) and all \(s\in\Gamma(\bbS)\),
	\item Leibniz rule with respect to the Clifford action:
	\begin{equation}
		\eL_X (\Phi \cdot s) = (\eL_X\Phi)\cdot s + \Phi \cdot (\eL_X s)
	\end{equation}
	for all Killing vectors \(X\) and all \(\Phi \in \Omega^\bullet(M)\), \(s\in\Gamma(\bbS)\),
	\item Compatibility with the Levi-Civita connection:
	\begin{equation}
		\eL_X \nabla_Y s = \nabla_{\eL_X Y}s + \nabla_Y \eL_X s
	\end{equation}
	for all Killing vectors \(X\) and all \(Y\in\fX(M)\), \(s\in\Gamma(\bbS)\).
\end{enumerate}
\end{lemma}

We have an additional result for backgrounds.

\begin{lemma}\label{lemma:liedersuperD}
If \((M,g,C,F)\) is a 5-dimensional background with superconnection \(\eD\),
\begin{equation}
	\eL_X \eD_Y s = \eD_{\eL_X Y}s + \eD_Y \eL_X s - (\eL_X\beta)_Y s
\end{equation}
for all Killing vectors \(X\) and all \(Y \in \fX(M)\), \(s\in\Gamma(\bbS)\).
\end{lemma}

\begin{proof}
Using Lemma~\ref{prop:liederprops} and the definition of \(\eD\),
\begin{equation}
\begin{split}
	\eL_X \eD_Y s
		&= \eL_X \nabla_Y s - \eL_X (\beta_Y s)\\
		&= \nabla_{\eL_X Y}s + \nabla_Y \eL_X s - (\eL_X\beta)_Y s - \beta_{\eL_X Y} s - \beta_Y\eL_X s \\
		&= \eD_{\eL_X Y} s + \eD_Y \eL_X s - (\eL_X\beta)_Y s,
\end{split}
\end{equation}
hence the result.
\end{proof}

\subsection{Spinor field bilinears}
\label{sec:spinor-field-bilinears}

We now "geometrise" our definitions of the spinor bilinears equation~\eqref{eq:spinor-bilinears}. For \(s\in\Gamma(\bbS)\), we define \(\mu_s\in C^\infty(M)\), \(\kappa_s\in\fX(M)\) and \(\omega_s^{AB}\in\Omega^2(M;\CC)\) as follows:
\begin{align}
	\mu_s &= \langle s,s \rangle,
	& g(\kappa_s,X) &= \langle s, X\cdot s \rangle,
	& \omega_s^{AB}(X,Y) &=\tfrac{1}{2} \bar{s}^A \comm{X}{Y}\cdot s^B,
\end{align}
for all \(X,Y\in\fX(M)\). It is then useful to define the (linear) Dirac current map \(\kappa:\odot^2\Gamma(\bbS)\to \fX(M)\) by
\begin{equation}
	g(\kappa(s_1,s_2),X) = \langle s_1, X\cdot s_2 \rangle,
\end{equation}
and we note that clearly \(\kappa_s=\kappa(s,s)\).

\begin{lemma}\label{lemma:Liederequi}
The Dirac current map is equivariant under the action of Killing vector fields via \(\eL\);
\begin{equation}
	\eL_X \kappa(s,s) = 2\kappa(\eL_X s,s)
\end{equation}
for all spinor fields \(s\in\Gamma(\bbS)\) and Killing vector fields \(X\).
\end{lemma}

\begin{proof}
First, note that since \(\kappa\)  is symmetric in its arguments, it is sufficient to prove the result on the diagonal since the full result then follows by a polarisation identity. Hence, \(X\) be a Killing vector field, \(Y\) a vector field and \(s\) a spinor field. Then locally,
\begin{equation}
\begin{split}
	Y\cdot ((\nabla X) s) 
		&= -\tfrac{1}{4} Y^\lambda \nabla_\mu X_\nu \Gamma_\lambda \Gamma^{\mu\nu} s\\
		&= -\tfrac{1}{4} Y_\lambda \nabla_\mu X_\nu \Gamma^{\lambda\mu\nu}s
			- \tfrac{1}{2} Y^\mu \nabla_\mu X_\nu \Gamma^\nu s\\
		&= -\tfrac{1}{8} \epsilon^{\lambda\mu\nu\sigma\tau}Y_\lambda \nabla_\mu X_\nu \Gamma_{\sigma\tau}\cdot s
			- \tfrac{1}{2} (\nabla_Y X)\cdot s
\end{split}
\end{equation}
and so, since \(\left <s,\Gamma_{\sigma\tau}s\right >\) = 0,
\begin{equation}
	\left<s,Y\cdot(\nabla X)s\right> 
		= -\tfrac{1}{2} \left<s,(\nabla_Y X)\cdot s\right>.
\end{equation}
Using this and the definition of the Dirac current, we have
\begin{equation}
\begin{split}
	g(\eL_X \kappa(s,s),Y)
		&=X \qty(g(\kappa(s,s),Y)) - g(\kappa(s,s),\eL_X Y)\\
		&=g(\nabla_X\kappa(s,s),Y) + g(\kappa(s,s),\nabla_X Y) - g(\kappa(s,s),\comm{X}{Y})\\
		&= 2g(\kappa(s,\nabla_X s),Y) + g(\kappa(s,s),\nabla_Y X)\\
		&= \left<s,2 Y\cdot(\nabla_X s) + (\nabla_Y X)\cdot s\right>\\
		&= \left<s,2 Y\cdot \qty(\nabla_X s - (\nabla X)s)\right>\\
		&= 2\left<s,Y\cdot \eL_X s\right>\\
		&= 2g(\kappa(s,\eL_X s)),Y),
\end{split}
\end{equation}
hence the result, since \(Y\) is arbitrary. 
\end{proof}

We also ``geometrise" the \(\gamma\) component of a Spencer cocycle as follows. The map \(\gamma:\odot^2 S\to\fso(TM)\) is defined via the first cocycle condition~\eqref{eq:cocycle-1}:
\begin{equation}
	\gamma(s,s)X = -2\kappa(s,\beta_Xs)
\end{equation}
for all \(s\in\Gamma(\bbS)\) and \(X\in\fX(M)\). In components, this is given by \eqref{eq:spencer-22-cocycle-solution}.

\subsection{Properties of Killing spinors}
\label{sec:killing-spinor-props}

From now on, we work with a background \((M,g,C,F)\). We are now almost ready to define the Killing superalgebra of such a background. However, we will see that in order for its bracket to close, the Dirac current of a Killing spinor must be a Killing vector which preserves the background fields \(C\) and \(F\). The following proposition will help to show this. 

\begin{proposition}
Let \(s\in\Gamma(\bbS)\) be a Killing spinor field, and for notational convenience fix \(\mu,\kappa,\omega\) to be the Dirac bilinears of \(s\). Then the following identities hold:
\begin{align}
	\label{eq:der-mu}
	\nabla_\mu \mu 
		&= 2C_{\mu\alpha}\kappa^\alpha 
			+ \tfrac{1}{2} F\indices{^\alpha_{AB}} \omega\indices{_{\mu\alpha}^{AB}}
\\
	\label{eq:der-kappa}
	\nabla_\mu \kappa_\nu 
		&= \gamma_{\mu\nu}(s,s)
		= 2C_{\mu\nu} \mu
			+\tfrac{1}{2}\epsilon_{\mu\nu\alpha\beta\gamma}C^{\alpha\beta}\kappa^\gamma
			+ \tfrac{1}{4}\epsilon_{\mu\nu\alpha\beta\gamma}F\indices{^\alpha_{AB}}
				\omega\indices{^{\beta\gamma}^{AB}}
\\
	\label{eq:der-omega}
	\begin{split}
	\nabla_\mu\omega_{\nu\rho}^{AB} 
		&= \epsilon_{\alpha\beta\gamma\mu[\nu}C^{\alpha\beta}
			\omega\indices{^\gamma_{\rho]}^{AB}} 
		+\epsilon_{\alpha\beta\gamma\nu\rho}
			C\indices{^\alpha_\mu}\omega\indices{^{\beta\gamma}^{AB}}
		-\tfrac{1}{4}\epsilon_{\mu\nu\rho\alpha\beta}F\indices{^{\alpha}^{AB}}\kappa^{\beta}
		+\tfrac{1}{2}\eta_{\mu[\nu}F\indices{_{\rho]}^{AB}}\mu
	\\
		& \qquad\qquad\qquad
		+ F\indices{_{[\nu}^{(A}_C}\omega\indices{_{\rho]\mu}^{B)C}} 
		- F\indices{_{\mu}^{(A}_C}\omega\indices{_{\nu\rho}^{B)C}}
		- g\indices{_{\mu[\nu}} F\indices{^\alpha^{(A}_C}\omega\indices{_{\rho]\alpha}^{B)C}},
	\end{split}
\end{align}
and in particular
\begin{align}\label{eq:div-omega}
	\nabla^\nu\omega_{\nu\rho}^{AB} 
		&= \tfrac{1}{2}\epsilon_{\rho\alpha\beta\gamma\delta}C^{\alpha\beta}
			\omega\indices{^{\gamma\delta}^{AB}}
		+F\indices{_{\rho}^{AB}}\mu
		- \tfrac{1}{2}F\indices{^\nu^{(A}_C}\omega\indices{_{\rho\nu}^{B)C}}
\\ \label{eq:ext-omega}
	\nabla_{[\mu}\omega_{\nu\rho]}^{AB} 
		&= -\tfrac{1}{4}\epsilon_{\mu\nu\rho\alpha\beta}
			F\indices{^{\alpha}^{AB}}\kappa^{\beta}.
\end{align}
\end{proposition}

\begin{proof}
Each of these identities follow from the Killing spinor equation. We have
\begin{equation}
	\nabla_\mu (\bar{s}^As^B)
	 	= \overline{\nabla_\mu s}^As^B + \bar{s}^A\nabla_\mu s^B
		= - \bar{s}^B\nabla_\mu s^A + \bar{s}^A\nabla_\mu s^B
		= 2\bar{s}^{[A}\nabla_\mu s^{B]}
		= 2\bar{s}^{[A}\qty(\beta_\mu s)^{B]},
\end{equation}
which is equivalent to
\begin{equation}
	\nabla_\mu \mu = 2\epsilon_{AB}2\bar{s}^{A}\qty(\beta_\mu s)^{B},
\end{equation}
and substituting in equation~\eqref{eq:beta-comps} for \(\beta\) yields equation~\eqref{eq:der-mu}. Similarly, we have
\begin{equation}
	\nabla_\mu \kappa_\nu = 2\epsilon_{AB}\bar{s}^{A}\Gamma_\nu\qty(\beta_\mu s)^{B},
		\label{eq:nabla-kappa}
\end{equation}
which, using cocycle condition~\eqref{eq:cocycle-1} gives equation~\eqref{eq:der-kappa}, and
\begin{equation}
	\nabla_\mu \omega_{\nu\rho}^{AB}
		= 2\bar{s}^{A}\Gamma_{\nu\rho}\qty(\beta_\mu s)^{B},
		\label{eq:nabla-omega}
\end{equation}
which after expanding and evaluating the resulting products of \(\Gamma\)-matrices yields equation~\eqref{eq:der-omega}. The expression \(\nabla\omega^{AB}\) has three components -- the trace, skew-symmetrisation and the elbow -- which we will treat separately. The trace part gives the divergence of \(\omega^{AB}\), while the skew-symmetrisation gives its exterior derivative
\begin{equation}
\begin{split}
	\nabla_{[\mu}\omega_{\nu\rho]}^{AB} &= 
	\epsilon_{\alpha\beta\gamma[\mu\nu}C^{\alpha\beta}\omega\indices{^\gamma_{\rho]}^{AB}} 
	+\epsilon_{\alpha\beta\gamma[\mu\nu}
		C\indices{^\alpha_{\rho]}}\omega\indices{^{\beta\gamma}^{AB}}
	-\tfrac{1}{4}\epsilon_{\mu\nu\rho\alpha\beta}F\indices{^{\alpha}^{AB}}\kappa^{\beta},
\end{split}
\end{equation}
of which the first two terms cancel identically by equation~\eqref{eq:tensor-id-AB-2}, giving equation~\eqref{eq:ext-omega}.
\end{proof}

We highlight one of these identities in particular.

\begin{corollary}\label{coro:kappa-gradient}
If \(s\) is a Killing spinor field, the gradient of its Dirac current is given by
\begin{equation}
	\nabla \kappa = - \gamma(s,s).
\end{equation}
In particular, \(\kappa\) is a Killing vector.
\end{corollary}

Note that this is a direct consequence of the cocycle condition equation~\eqref{eq:cocycle-1}.

\subsection{Existence of Killing superalgebras}
\label{sec:killing-exist}

We can now construct the supersymmetry algebra of a supersymmetric background \((M,g,C,F)\). For such a background, we denote the space of (symplectic Majorana) Killing spinors by \(\fK_{\bar{1}}\); that is,
\begin{equation}
	\fK_{\bar{1}} = \qty{s \in \Gamma(\bbS) |\, \eD s = 0},
\end{equation}
and we also define
\begin{equation}
	\fK_{\bar{0}} = \qty{X \in \fX(M) |\, \eL_X g = \eL_X C = \eL_X F = 0},
\end{equation}
the space of Killing vector fields which preserve \(C\) and \(F\). We define the bracket \(\comm{\cdot}{\cdot}\) on \(\fK = \fK_{\bar{0}}\oplus\fK_{\bar{1}}\) by extension of the following:
\begin{itemize}
	\item the usual Lie bracket of vector fields on \(\fK_{\bar{0}}\otimes\fK_{\bar{0}}\), \(\comm{X}{Y} = \eL_XY\),
	\item the Dirac current on \(\fK_{\bar{1}}\otimes\fK_{\bar{1}}\), \(\comm{s}{s} = \kappa_s\),
	\item the spinorial Lie derivative on \(\fK_{\bar{0}}\otimes\fK_{\bar{1}}\), \(\comm{X}{s} = \eL_X s\). 
\end{itemize}

\begin{theorem}\label{thm:killing-superalg-exist}
If \((M,g,C, F)\) is a background such that \(\Gamma^\nu{R^\eD}_{\mu\nu}=0\),
then \((\fK,\comm{\cdot}{\cdot})\) as described above is a Lie superalgebra.
\end{theorem}

\begin{proof}
Throughout, let \(s\) be a Killing spinor field and let \(\mu, \kappa\) and \(\omega^{AB}\) be the bilinears associated to \(s\). First, note that the super skew-symmetry of the bracket follows directly from the definition. We will next show that \(\fK\) is closed under the bracket operation. First, we have \(\eL_\kappa g=0\) by Corollary~\ref{coro:kappa-gradient}. Using the Cartan formula for the Lie derivative, we have
\begin{equation}
	\eL_\kappa C = d\iota_\kappa C + \iota_\kappa dC,
\end{equation}
but by equation~\eqref{eq:der-mu}, \(\iota_\kappa C = -\tfrac{1}{2}d\mu-\tfrac{1}{4}\iota_{F_{AB}}\omega^{AB}\), so
\begin{equation}
	\eL_\kappa C = -\tfrac{1}{4}d \iota_{F_{AB}}\omega^{AB}+\iota_\kappa dC.
\end{equation}
In components, the first term is proportional to
\begin{equation}\label{eq:d-contract-f-omega}
\begin{split}
	\nabla\indices{_{[\mu}}\qty(\omega\indices{_{\nu]\rho}^{AB}}F\indices{^\rho_{AB}})
		&= \qty(\nabla\indices{_{[\mu}}\omega\indices{_{\nu]\rho}^{AB}})
			F\indices{^\rho_{AB}}
		+\omega\indices{_{[\nu|\rho|}^{AB}}
			\qty(\nabla\indices{_{\mu]}}F\indices{^\rho_{AB}})\\
		&= \tfrac{3}{2}\qty(\nabla\indices{_{[\mu}}\omega\indices{_{\nu\rho]}^{AB}})
			F\indices{^\rho_{AB}}
		- \tfrac{1}{2}\qty(\nabla\indices{_\rho}\omega\indices{_{\mu\nu}^{AB}})
			F\indices{^\rho_{AB}}
		+\omega\indices{_{[\nu|\rho|}^{AB}}
			\qty(\nabla\indices{_{\mu]}}F\indices{^\rho_{AB}}).
\end{split}
\end{equation}
Using equation~\eqref{eq:ext-omega}, the first term in this expansion vanishes identically:
\begin{equation}
	\tfrac{3}{2}\qty(\nabla\indices{_{[\mu}}\omega\indices{_{\nu\rho]}^{AB}})
			F\indices{^\rho_{AB}}
		= -\tfrac{3}{8}\epsilon_{\mu\nu\rho\alpha\beta}
			\qty(F^\rho \cdot F^{\alpha})\kappa^{\beta}
		= 0,
\end{equation}
because \(F^\rho \cdot F^{\alpha}\) is symmetric. Substituting equation~\eqref{eq:der-omega} into the second term of equation~\eqref{eq:d-contract-f-omega}, three terms vanish by symmetry, leaving
\begin{equation}
\begin{split}
	-\tfrac{1}{2}\qty(\nabla_\rho\omega_{\mu\nu}^{AB})F^\rho_{AB}
		&= -\tfrac{1}{2}\qty[
			\epsilon_{\alpha\beta\gamma\rho[\mu}C^{\alpha\beta}F^\gamma_{AB}
				\omega\indices{_{\nu]}^\rho^{AB}} 
			+\epsilon_{\mu\nu\alpha\beta\gamma}C\indices{^\alpha_\rho}F^\rho_{AB}
				\omega\indices{^{\beta\gamma}^{AB}}
			+ \comm{F^\rho}{F_{[\mu}}_{BC}\omega_{\nu]\rho}^{BC}
		].
\end{split}
\end{equation}
Thus 
\begin{equation}
\begin{split}
	\nabla\indices{_{[\mu}}\qty(\omega\indices{_{\nu]\rho}^{AB}}F\indices{^\rho_{AB}})
		=& \qty(\nabla_{[\mu}F^\rho_{AB} 
			-\tfrac{1}{2}\epsilon_{\alpha\beta\gamma\rho[\mu}C^{\alpha\beta}F^\gamma_{AB}
			)\omega_{\nu]\rho}^{AB}	\\
		&-\tfrac{1}{2}\epsilon_{\mu\nu\alpha\beta\gamma}C\indices{^\alpha_\rho}F^\rho_{AB}
				\omega\indices{^{\beta\gamma}^{AB}}
			-\tfrac{1}{2}\comm{F^\rho}{F_{[\mu}}_{BC}\omega_{\nu]\rho}^{BC}.
\end{split}
\end{equation}
Thus \(\eL_\kappa C\) vanishes when \(\Gamma^\nu{R^\eD}_{\mu\nu}=0\) by Theorem~\ref{thm:gamma-trace-zero}. Since each \(F\indices{_{AB}}\) is a one-form, we can also use the Cartan formula for its Lie derivative:
\begin{equation}
	\eL_\kappa F_{AB} = d\iota_\kappa F_{AB} + \iota_\kappa d F_{AB},
\end{equation}
which in components gives
\begin{equation}
	(\eL_\kappa F_{AB})_\mu
	=\nabla_\mu(\kappa^\nu F\indices{_\nu_{AB}}) 
		+ \kappa^\nu \nabla_\nu F\indices{_\mu_{AB}}
		- \kappa^\nu \nabla_\mu F\indices{_\nu_{AB}}
	= (\nabla_\mu\kappa^\nu) F\indices{_\nu_{A B}}
		+ \kappa^\nu\nabla_\nu F\indices{_{\mu A B}}.
\end{equation}
Using equation~\eqref{eq:der-kappa}, the first term on the right hand side is
\begin{equation}
\begin{split}
	\gamma_{\mu\nu}F\indices{^\nu_{A B}}
	&=2C_{\mu\nu}F\indices{^\nu_{A B}}\mu
		+\tfrac{1}{2}\epsilon_{\mu\nu\alpha\beta\gamma}
			C^{\alpha\beta}F\indices{^\nu_{A B}}\kappa^\gamma
		+ \tfrac{1}{4}\epsilon_{\mu\nu\alpha\beta\gamma}
			F\indices{^\nu_{A B}}F\indices{^\alpha_{CD}}
			\omega\indices{^{\beta\gamma}^{CD}}\\
	&=2C_{\mu\nu}F\indices{^\nu_{A B}}\mu
		+\tfrac{1}{2}\epsilon_{\mu\alpha\beta\gamma\delta}
			C^{\alpha\beta}F\indices{^\gamma_{A B}}\kappa^\delta
		+ \tfrac{1}{8}\epsilon\indices{_{\mu\alpha\beta\sigma\tau}}
			\comm{F^\alpha}{F^\beta}\indices{^C_{(A}}\omega\indices{^{\sigma\tau}_{B)}_C},
\end{split}
\end{equation}
where the second equality arises as follows: we have
\begin{align} 
	F\indices{^{[\alpha}_{A(B}}F\indices{^{\beta]}_{C)D}}
	&=F\indices{^{[\alpha}_{[A|(B}}F\indices{^{\beta]}_{C)|D]}}
	=\tfrac{1}{2}\epsilon_{AD}\epsilon^{EF}F\indices{^{[\alpha}_{E(B}}F\indices{^{\beta]}_{C)F}}
	=-\tfrac{1}{4}\epsilon_{AD}\comm{F^\alpha}{F^\beta}_{BC},\\
	F\indices{^{[\alpha}_{A[B}}F\indices{^{\beta]}_{C]D}}
	&=\tfrac{1}{2}\epsilon_{BC}\epsilon^{EF}F\indices{^{[\alpha}_{AE}}F\indices{^{\beta]}_{FD}}
	=-\tfrac{1}{4}\epsilon_{BC}\comm{F^\alpha}{F^\beta}_{AD},
\end{align}
and so
\begin{equation}
	F\indices{^{[\alpha}_{AB}}F\indices{^{\beta]}_{CD}} 
		= -\tfrac{1}{4}\qty(\epsilon_{AD}\comm{F^\alpha}{F^\beta}_{BC} + \epsilon_{BC}\comm{F^\alpha}{F^\beta}_{AD}),
\end{equation}
and finally symmetrising in \(AB\) gives \(F\indices{^{[\alpha}_{AB}}F\indices{^{\beta]}_{CD}} = \tfrac{1}{2}\epsilon\indices{_{(C|(A}} \comm{F^\alpha}{F^\beta}\indices{_{B)|D)}}\)\footnote{This reflects the Hodge isomorphism \(\wedge^2\odot^2\Delta\cong\odot^2\Delta\) with respect to the inner product on \(\odot^2\Delta\) induced by the symplectic product on \(\Delta\).}. We thus have
\begin{equation}
	\qty(\eL_\kappa F_{AB})_\mu 
	= 2C_{\mu\nu}F\indices{^\nu_{A B}}\mu
		+\qty(
			\nabla_\delta F\indices{_\mu_{AB}}
			+\tfrac{1}{2}\epsilon_{\mu\alpha\beta\gamma\delta}
				C^{\alpha\beta}F\indices{^\gamma_{A B}}
			)\kappa^\delta
		+ \tfrac{1}{4}\epsilon\indices{_{\mu\alpha\beta\sigma\tau}}
			\comm{F^\alpha}{F^\beta}\indices{^C_{(A}}\omega\indices{^{\sigma\tau}_{B)}_C}.
\end{equation}

By Theorem~\ref{thm:gamma-trace-zero}, this vanishes if \(\Gamma^\nu{R^\eD}_{\mu\nu}=0\). We have thus shown that \(\comm{\fK_{\bar{1}}}{\fK_{\bar{1}}} \subseteq {\fK_{\bar{0}}}\). One can easily verify that \(\comm{\fK_{\bar{0}}}{\fK_{\bar{0}}} \subseteq {\fK_{\bar{0}}}\). We must show that \(\comm{\fK_{\bar{1}}}{\fK_{\bar{0}}} + \comm{\fK_{\bar{0}}}{\fK_{\bar{1}}} \subseteq {\fK_{\bar{1}}}\); that is, \(\eL_K s\in\fK_{\bar{1}}\) for all \(K\in\fK_{\bar{0}}\) and \(s\in\fK_{\bar{1}}\). By Lemma~\ref{lemma:liedersuperD}, for all \(X\in\fX(M)\) we have
\begin{equation}
	\eD_X \eL_K s = \eD_{\comm{X}{K}}s + \eL_K\eD_X s + (\eL_K\beta)_Xs = (\eL_K\beta)_Xs,
\end{equation}
where we have used \(\eD s=0\), but using the expression~\eqref{eq:beta-clifford} for \(\beta\) and the Leibniz rule for the Lie derivative with respect to Clifford multiplication, we see that \(\eL_K\beta=0\) for all \(K\in\fK_{\bar{0}}\). Thus \(\eL_K s\in\fK_{\bar{1}}\).

It remains to be shown only that the Jacobi identity is satisfied. For
three Killing vector fields this is clear because the bracket is simply
the commutator. For two Killing vector fields \(X,Y\) and a Killing
spinor field \(s\),
\begin{equation}
	\comm{s}{\comm{X}{Y}} + \comm{X}{\comm{Y}{s}} - \comm{Y}{\comm{s}{X}}
		= -\eL_{\eL_XY}s + \eL_X \eL_Y s - \eL_Y \eL_X s,
\end{equation}
so the Jacobi identity follow from \(\eL\) being a Lie algebra
representation of vector fields on spinor fields, which follows from
Proposition~\ref{prop:liederprops}. By symmetry, for the case of one
Killing vector field and two Killing spinor fields, we need only
consider the identity where both spinor fields are the same:
\begin{equation}
	\comm{X}{\comm{s}{s}} + \comm{s}{\comm{s}{X}} - \comm{s}{\comm{X}{s}}
		= \comm{X}{\kappa(s,s)} - \kappa(s,\eL_X s) - \kappa(s,\eL_X s),
\end{equation}
which vanishes by Lemma~\ref{lemma:Liederequi}. Finally, for three
Killing spinor fields, again we need only consider the case where they
are all the same; using Corollary~\ref{coro:kappa-gradient} this reduces
to the vanishing of
\begin{equation}
	\comm{\comm{s}{s}}{s} 
		= \eL_\kappa s
		= \nabla_\kappa s - (\nabla\kappa) s
		= \beta_\kappa s + \gamma(s,s) s,
\end{equation}
which is simply the cocycle condition~\eqref{eq:cocycle-2}.
\end{proof}

We can now finally state the definition of the Killing superalgebra.

\begin{definition}
The \textbf{Killing superalgebra} of a background \((M,g,C, F)\) with \(\Gamma^\nu{R^\eD}_{\mu\nu}=0\) is the Lie superalgebra \(\fK = \fK_{\bar{0}}\oplus\fK_{\bar{1}}\) where
\begin{equation}
	\fK_{\bar{0}} = \qty{X \in \fX(M) |\, \eL_X g = \eL_X C = \eL_X F = 0}
\end{equation}
is the space of Killing vector fields which preserve \(C\) and \(F\),
\begin{equation}
	\fK_{\bar{1}} = \qty{s \in \Gamma(\bbS) |\, \eD s = 0}
\end{equation}
is the space of Killing spinor fields, and the bracket \(\comm{\cdot}{\cdot}\) is defined above.
\end{definition}

In the proof of Theorem~\ref{thm:killing-superalg-exist}, the vanishing 
Clifford trace of the curvature is used to show that the algebra 
closes. One might ask whether this is necessary; that is,
do Killing superalgebras only exist for backgrounds which satisfy the
equations of motion? Recalling Theorem~\ref{thm:gamma-trace-zero}, note
that equations~\eqref{eq:sugra-maxwell} and \eqref{eq:sugra-einstein}
were not necessary to close the algebra. Thus, if \(F=0\), \(dC=0\) is
sufficient to close the algebra. This is the case in 5-dimensional
supergravity: \(C\) is (proportional to) the field strength of a Maxwell
field and hence closed, so there may exist "off-shell" supersymmetric
bosonic backgrounds. On the other hand, if \(C=0\), \(\nabla F=[F,F]=0\)
suffices.

\section{Maximally supersymmetric backgrounds}
\label{sec:max-susy}

Throughout we work in a maximally supersymmetric background
$(M,g,C,F)$ which we assume to be connected.  We seek only to classify
the geometries up to local isometry, so we also assume that $M$ is
simply connected.  Under this assumption, maximal supersymmetry is
equivalent to the vanishing of the curvature $R^\eD$ of the
superconnection, so Theorem~\ref{thm:maximalsusy} gives necessary
\emph{and} sufficient conditions for maximal supersymmetry. We now work
with local coordinate frames rather than local orthonormal frames, still
using Greek indices.

\subsection{Maximally supersymmetric supergravity backgrounds}
\label{sec:max-susy-sugra}

Taking $F=0$ reduces the problem to the determination of maximally
supersymmetric geometries in minimal 5-dimensional supergravity. These
are already known in the literature.  They were constructed directly in 
\cite{Gauntlett:2002nw} and via quotients from the maximally supersymmetric
backgrounds of $(1,0)$ six-dimensional supergravity in
\cite{Chamseddine:2003yy}.  They are given by
\begin{itemize}
\item the near-horizon geometry of the BMPV black hole \cite{Breckenridge:1996is},
\item $AdS_3\times S^2$ and $AdS_2\times S^3$, which also arise as
  limits of the BMPV near-horizon solution,
\item a particular Cahen--Wallach pp-wave \cite{Meessen:2001vx},
\item a 5-dimensional analogue of the Gödel universe.
\end{itemize}
In addition, in \cite{Gauntlett:2002nw} there are three further candidate
geometries which, as the authors already speculate, are in fact locally
isometric to cases already listed above.  This is easy to check by
calculating their Killing superalgebras and showing that they are
isomorphic to the ones above and observing, as was shown in
\cite{Figueroa-OFarrill:2016khp} in the context of eleven-dimensional
supergravity but holds more generally, that for a ($>\tfrac12$)-BPS
background, the Killing superalgebra, which is transitive, determines
the geometry up to local isometry.

\subsection{Maximally supersymmetric backgrounds with $F\neq 0$}
\label{sec:max-susy-rigid}

By Theorem~\ref{thm:maximalsusy}, the geometry is given by a nonzero
parallel one-form $\varphi$ (which we now regard dually as a vector
field) and the Riemann curvature is constrained by
equation~\eqref{eq:max-susy-riemann-F}, which in a coordinate frame
reads
\begin{equation}\label{eq:riem-non-sugra}
  R_{\mu\nu\sigma\tau}
  = \varphi^2 g_{\mu[\sigma} g_{\tau]\nu} - \varphi_\mu
  \varphi_{[\sigma} g_{\tau]\nu} + \varphi_\nu \varphi_{[\sigma}
  g_{\tau]\mu}.
\end{equation}
The Ricci and scalar curvatures are then
\begin{equation}
    R_{\mu\nu} = \tfrac{3}{2}\qty(\varphi_\mu\varphi_\nu-\varphi^2g_{\mu\nu}) \qquad\text{and}\qquad R = -6\varphi^2,
\end{equation}
and the Weyl tensor vanishes, so the metric is conformally flat. It is
also locally symmetric (hence symmetric, since $M$ is simply connected):
$\nabla_\lambda R_{\mu\nu\sigma\tau}=0$ since $\nabla\varphi=0$. Another
consequence of $\varphi\neq 0$ being parallel is that it is nowhere
vanishing and $\varphi^2$ is constant -- in particular, the scalar
curvature is constant. The geometry is thus determined by the causal
type of $\varphi$.

\subsubsection{Spacelike $\varphi$}

When $\varphi^2<0$, $\varphi$ defines a distribution of rank 1 in $TM$
which is preserved by the holonomy of $\nabla$ since $\varphi$ is
parallel. The rank-4 perpendicular distribution is also preserved by
holonomy, and the metric is nondegenerate on either distribution. The de
Rham--Wu decomposition theorem then allows us to decompose $(M,g)$ as a
product $(N,g_N)\times(\RR,-\dd x^2)$, where $(N,g_N)$ is a
4-dimensional lorentzian manifold and $x$ is the standard coordinate on
$\RR$. We then define the symmetric tensor
\begin{equation}
  h_{\mu\nu} = g_{\mu\nu} - \frac{\varphi_\mu\varphi_\nu}{\varphi^2}
\end{equation}
so we can write the curvature tensors as
\begin{equation}
  \begin{split}
    R_{\mu\nu\sigma\tau} &= \varphi^2 h_{\mu[\sigma}h_{\tau]\nu}\\
   R_{\mu\nu} &= -\tfrac{3}{2}\varphi^2 h_{\mu\nu}\\
  R &= -6\varphi^2.
\end{split}
\end{equation}
The pullback of $h$ to $N$ coincides with $g_N$. Since $(N,g_N)$ is a
lorentzian symmetric space with constant positive scalar curvature \footnote{One might object that $AdS_4$ has \emph{negative} scalar curvature; however, this is only true in mostly-positive signature. The Ricci tensor is invariant under the homothety \(g \to -g\), while the scalar curvature undergoes a change of sign, thus \(AdS_4\) has positive curvature in our conventions.}, it must be $AdS_4$ by a result of Cahen and Wallach
\cite{MR267500,MR461388}.

\subsubsection{Timelike $\varphi$}

We can treat this similarly to the spacelike case. Here, we get a
decomposition $(\RR,\dd t^2)\times (N,-g_N)$ where $t$ is the standard
coordinate on $R$ and $(N,g_N)$ is a riemannian manifold. Note that the
pullback of $h$ coincides with $-g_N$. The sectional curvature of
$(N,g_N)$ is a positive constant, so it is $S^4$.

\subsubsection{Null $\varphi$}

In this case, the de Rham--Wu theorem cannot be used since the
distribution defined by $\varphi$ is degenerate. The geometry here is a
Brinkmann pp-wave space -- a lorentzian manifold with a parallel null
vector field. The curvature tensors reduce to
\begin{equation}\label{eq:null-varphi-curvature}
  \begin{split}
    R_{\mu\nu\sigma\tau} &=
    -g\indices{_{\mu[\sigma}}\varphi_{|\nu|}\varphi_{\tau]} + g\indices{_{\nu[\sigma}}\varphi_{|\mu|}\varphi_{\tau]}\\
    R_{\mu\nu} &= \tfrac{3}{2}\varphi_\mu\varphi_\nu
\end{split}
\end{equation}
and $R=0$.  Thus $(M,g)$ is a scalar-flat lorentzian symmetric space; by
the Cahen--Wallach theorem \cite{MR267500,MR461388} it is a
Cahen--Wallach pp-wave $CW_5(A)$. Such a space has a coordinate system
$(x^+,x^-,x^1,x^2,x^3)$ in which the metric is given by
\begin{equation}
  g = 2\dd x^+\dd x^- - \qty(\sum_{i,j=1}^3 A_{ij}x^ix^j)(\dd x^-)^2 -
  \sum_{i=1}^3 \qty(\dd x^i)^2,
\end{equation}
where $A\in \odot^2\RR^3$. We take $\varphi=\partial_+$ to be the
distinguished parallel null vector field. This metric is scalar-flat for
any $A$, and the non-vanishing components of the Riemann and Ricci
tensors are $R_{i-j-}=-A_{ij}$ and
$R_{--}=-\sum_{i=1}^3A_{ii}$. The non-vanishing components of
the Weyl tensor are given by the trace-free part of $A$: namely,
\begin{equation}
  W_{i-j-} = -A_{ij} + \tfrac{1}{3}\delta_{ij}\sum_{k=1}^3A_{kk}.
\end{equation}
We already saw that the Weyl tensor for a maximally supersymmetric
geometry vanishes, so we have $A_{ij}=a\delta_{ij}$ where
$a=\tfrac{1}{3}\sum_{k=1}^3A_{kk}$. Now, comparing
with equations~\eqref{eq:null-varphi-curvature} for \(\varphi=\partial_+\), we find that
$a=-\frac{1}{2}$. We have thus shown the following.

\begin{theorem}\label{thm:max-susy-bgs}
  Let $(M,g,C,F)$ be a maximally supersymmetric 5-dimensional
  background. If $F=0$ then $(M,g)$ is a maximally supersymmetric
  background of minimal 5-dimensional supergravity. If $F\neq 0$ then
  $C=0$ and $F=\varphi\otimes r$ for some one-form $\varphi$ and some
  $r\in\fsp(1)$, and up to local isometry,
  \begin{itemize}
  \item If $\varphi^2>0$, $(M,g)=\RR^{1,0}\times S^4$, where $S^4$ is
    the round sphere with scalar curvature $6\varphi^2$ (radius
    $\sqrt{\frac{2}{\varphi^2}}$),
  \item If $\varphi^2<0$, $(M,g)=AdS_4\times \RR^{0,1}$, where $AdS_4$ is the 4-dimensional
    anti-de Sitter spacetime with scalar curvature $6\varphi^2$
    (cosmological constant $\Lambda=\frac{3}{2}\varphi^2$),
  \item If $\varphi^2=0$, $(M,g)$ is a Cahen--Wallach pp-wave with coordinates
    $(x^+,x^-,x^1,x^2,x^3)$ in which $\varphi=\partial_+$ and
    \begin{equation}
      g = 2\dd x^+\dd x^- + \tfrac{1}{2}\qty(\sum_{i=1}^3 (x^i)^2)(\dd x^-)^2 - \sum_{i=1}^3 \qty(\dd x^i)^2.
    \end{equation}
  \end{itemize}
\end{theorem}

We note that this result here is (up to a rescaling of $\varphi$ and
$r$) precisely the 5-dimensional analogue of Theorem~27 part~(ii) in
\cite{deMedeiros:2018ooy}.

\subsection{Killing superalgebras}
\label{sec:kill-super}

In this section we will explicitly describe the Killing superalgebras of
the maximally supersymmetric backgrounds as filtered subdeformations of
the five-dimensional minimal Poincaré superalgebra.  As explained in
\cite{Figueroa-OFarrill:2016khp} in the context of eleven-dimensional
supergravity, the Killing superalgebra is generated by sections
of a supervector bundle $\eE = \eE_{\bar 0} \oplus \eE_{\bar 1}$, where
\begin{equation}
  \eE_{\bar 0} = TM \oplus \fso(TM) \qquad\text{and}\qquad \eE_{\bar 1} =
  \bbS,
\end{equation}
which are parallel relative to a superconnection $\eD$ which agrees on
sections of $\eE_{\bar 0}$ with the Killing transport connection
\cite{MR0084825,MR0250643} and on sections of $\eE_{\bar 1}$ with the
connection given by equation~\eqref{eq:superconnection}.  There might be
in addition additional (tensorial) constraints on the sections of
$\eE$.

For the case at hand, the Killing superalgebra is a filtered Lie
superalgebra whose underlying vector space is $V \oplus S \oplus \fh$,
where $\fh \subset \fso(V)$ is a subalgebra.  The Lie brackets of the
Killing superalgebra are defined by
\begin{equation}\label{eq:ksa}
  \begin{aligned}{}
    [A,B] &= AB - BA\\
    [A,s] &= \tfrac12 \omega_A \cdot s\\
    [A,v] &= Av + \underbrace{[A,\lambda_v] - \lambda_{Av}}_{\in\fh}
  \end{aligned}
  \qquad\qquad
  \begin{aligned}{}
    [s,s] &= \kappa(s) + \underbrace{\gamma^\Phi(s,s) - \lambda_{\kappa(s,s)}}_{\in\fh}\\
    [v,s] &= \beta^\Phi_v s + \tfrac12 \omega_{\lambda_v} \cdot s\\
    [v,w] &= \underbrace{\lambda_v w - \lambda_w v}_{\in V} + 
    \underbrace{R(v,w) + [\lambda_v,\lambda_w] - \lambda_{\lambda_v w -
        \lambda_w v}}_{\in\fh},
  \end{aligned}
\end{equation}
for all $A,B \in \fh$, $s \in S$ and $v,w \in V$. Here
$\beta^\Phi + \gamma^\Phi$ is the normalised cocycle (with $\alpha = 0$)
where $\Phi$ stands for the generic additional fields (here $C$ and $F$)
and $\fh = \fso(V) \cap \fstab(\Phi)$ is the Lie algebra of the
stabiliser of $\Phi$ in $\SO(V)$. In addition, $R$ is the Riemann
curvature tensor, induced from a map $\wedge^2 V \to \fso(V)$. Finally,
$\lambda : V \to \fso(V)$ is defined only up to a linear map
$V \to \fh$, so it is to be thought of more precisely as a linear map
$V \to \fso(V)/\fh$.

\subsubsection{Killing superalgebras for maximally supersymmetric supergravity backgrounds}
\label{sec:ksa-max-susy-sugra}

For these backgrounds, the normalised cocycle is given by
equation~\eqref{eq:spencer-22-cocycle-solution} with $F=0$.  In
components, we have
\begin{equation}
  \label{eq:norm-cocy-sugra}
  \begin{split}
    \beta^\Phi_\mu &= \tfrac18 C^{\sigma\tau} \left( \Gamma_\mu  \Gamma_{\sigma\tau}- 3  \Gamma_{\sigma\tau} \Gamma_\mu\right)\\
    \gamma^\Phi_{\mu\nu} &=2 \mu C_{\mu\nu} + \tfrac12 \epsilon_{\mu\nu\rho\sigma\tau}\kappa^\rho C^{\sigma\tau}.
  \end{split}
\end{equation}
Letting $\gamma(s,s) := \gamma^\Phi(s,s) - \lambda_{\kappa(s,s)}$, we
find that
\begin{equation}
  \gamma_{\mu\nu} = 2 \mu C_{\mu\nu} + \tfrac12
  \epsilon_{\mu\nu\rho\sigma\tau}\kappa^\rho C^{\sigma\tau} -
  \kappa^\rho \lambda_{\rho\mu\nu},
\end{equation}
where $\lambda_{\rho\mu\nu} = - \lambda_{\rho\nu\mu}$.  Demanding that
$\gamma(s,s) \in \fh$ is tantamount to demanding the vanishing of the Clifford
commutator
\begin{equation}
  [\gamma_{\mu\nu} \Gamma^{\mu\nu}, C_{\sigma\tau} \Gamma^{\sigma\tau}]
  = 0,
\end{equation}
since that is, up to inconsequential factors, the action of
$\gamma(s,s)$ on $C$.  It is clear that if we define $\lambda$ by
\begin{equation}
  \lambda_{\rho\mu\nu} = \tfrac12 \epsilon_{\rho\mu\nu\sigma\tau} C^{\sigma\tau},
\end{equation}
then $\gamma(s,s) \in \fh$, where $\gamma_{\mu\nu} = 2 \mu C_{\mu\nu}$.
We remark that $\lambda$ is $\fh$-equivariant, so that the $\fh$
component of the $[A,v]$ bracket in equation~\eqref{eq:ksa} is absent.

Defining $\beta_v := \beta_v^\Phi + \tfrac12 \omega_{\lambda_v}$, we see
that in components
\begin{equation}
  \begin{split}
    \beta_\mu &= \beta_\mu^\Phi + \tfrac14 \lambda_{\mu\sigma\tau} \Gamma^{\sigma\tau}\\
    &= \tfrac18 C^{\alpha\beta} \left( \Gamma_\mu  \Gamma_{\alpha\beta}- 3
      \Gamma_{\alpha\beta} \Gamma_\mu +
      \epsilon_{\mu\sigma\tau\alpha\beta}
      \Gamma^{\sigma\tau}\right)\\
    &= \tfrac14 C^{\alpha\beta} \epsilon_{\mu\alpha\beta\sigma\tau} \Gamma^{\sigma\tau} + C_{\mu\nu} \Gamma^\nu.
  \end{split}
\end{equation}
Now $\alpha(v,w) := \lambda_v w - \lambda_w v$ is given in components by
\begin{equation}
  \alpha_{\mu\nu\rho} = -\epsilon_{\mu\nu\rho\sigma\tau} C^{\sigma\tau}.
\end{equation}
Defining $\rho : \wedge^2 V \to \fso(V)$ by $\rho(v,w):= R(v,w) + [\lambda_v,\lambda_w] -
\lambda_{\alpha(v,w)}$, we can write the Lie brackets of the Killing
superalgebra as follows:
\begin{equation}\label{eq:ksa-sugra}
  \begin{aligned}{}
    [A,B] &= AB - BA\\
    [A,s] &= \tfrac12 \omega_A \cdot s\\
    [A,v] &= Av
  \end{aligned}
  \qquad\qquad
  \begin{aligned}{}
    [s,s] &= \kappa(s) + \gamma(s,s)\\
    [v,s] &= \beta_v s\\
    [v,w] &= \alpha(v,w) + \rho(v,w),
  \end{aligned}
\end{equation}
and determine $\rho$ by the Jacobi identity and check that it maps to
$\fh$.  Most of the components of the Jacobi identity are already
satisfied by construction:
\begin{itemize}
\item $[\fh,\fh,-]$ because $V \oplus S \oplus \fh$ is an $\fh$-module,
\item $[\fh,S,S]$ because $\gamma$ is $\fh$-equivariant,
\item $[\fh,S,V]$ because $\beta$ is $\fh$-equivariant,
\item $[S,S,S]$ because of the Spencer cocycle condition, and
\item the $V$-component of $[S,S,V]$ because of the Spencer cocycle
  condition.
\end{itemize}
We will use the $[S,V,V]$ Jacobi to define
$\rho$ and then check that $\rho$ is $\fh$-equivariant and maps to $\fh$
which means that the $[\fh,V,V]$ Jacobi is satisfied.  Finally, we have
to check that the $\fh$-component of the $[S,S,V]$ Jacobi as well as
$[V,V,V]$ Jacobi are satisfied.

The $[S,V,V]$ Jacobi says that for all $v,w \in V$ and $s \in S$,
\begin{equation}
  [v,[w,s]] - [w,[v,s]] \stackrel{!}{=} [[v,w],s]
\end{equation}
which is equivalent to
\begin{equation}
  \tfrac12 \omega_{\rho(v,w)} \stackrel{!}{=} [\beta_v,\beta_w] - \beta_{\alpha(v,w)}.
\end{equation}
This requires that the RHS should belong to $\wedge^2 V \subset \Cl(V)$,
which can be checked to be the case.  In components, the above condition
is
\begin{equation}
  \tfrac14 \rho_{\mu\nu\alpha\beta} \Gamma^{\alpha\beta} \stackrel{!}{=}
  [\beta_\mu, \beta_\nu] - \alpha_{\mu\nu}{}^\rho \beta_\rho
\end{equation}
and after a short calculation (which we omit) results in
\begin{equation}
  \rho_{\mu\nu\alpha\beta} = 4 C_{\mu\nu} C_{\alpha\beta}.
\end{equation}
It follows that $\rho$ is $\fh$-equivariant and moreover that it lands
in $\fh$.  Therefore the $[\fh,V,V]$ Jacobi is satisfied.  It is
straightforward, if somewhat tedious, to check that the rest of the Jacobi
identity is satisfied.

In summary, the Killing superalgebra of a maximally supersymmetric
background $(M,g,C)$ of minimal 5-dimensional supergravity is the
filtered Lie superalgebra with underlying vector space $\fg = V \oplus S
\oplus \fh$, where $\fh = \fso(V) \cap \fstab(C)$, whose brackets are
given for all $A,B \in \fh$, $s \in S$ and $v,w \in V$ by
equation~\eqref{eq:ksa-sugra}, with 
\begin{equation}
  \begin{split}
    \alpha_{\mu\nu\rho} &= - \epsilon_{\mu\nu\rho\sigma\tau} C^{\sigma\tau}\\
    \beta_\mu &= \tfrac14 \epsilon_{\mu\alpha\beta\sigma\tau} C^{\alpha\beta} \Gamma^{\sigma\tau} + C_{\mu\nu} \Gamma^\nu\\
    \gamma_{\mu\nu} &= 2 \mu C_{\mu\nu}\\
    \rho_{\mu\nu\alpha\beta} &= 4 C_{\mu\nu} C_{\alpha\beta}.
  \end{split}
\end{equation}

\subsubsection{Killing superalgebras for maximally supersymmetric backgrounds with $F \neq 0$}
\label{sec:ksa-max-susy-non-sugra}

For these backgrounds, $C=0$ and $F = \varphi \otimes r$ where $\varphi$
is a parallel vector field (or one-form) and $r \in \fsp(1)$ is a fixed
element of the R-symmetry Lie algebra.  The normalised cocycle
$\beta^\Phi + \gamma^\Phi$ can be read off from
equation~\eqref{eq:spencer-22-cocycle-solution}:
\begin{equation}
  \begin{split}
    \beta^\Phi(v,s)^A &= - \tfrac18 \left( v \cdot \varphi + 3 \varphi \cdot v \right) \cdot r^A{}_B s^B\\
    \gamma^\Phi(s,s)_{\mu\nu} &= \tfrac14 \epsilon_{\mu\nu\rho}{}^{\sigma\tau} \varphi^\rho r_{AB} \omega^{AB}_{\sigma\tau}.
  \end{split}
\end{equation}
It follows that $\gamma(s,s) \in \fh = \fso(V) \cap \fstab(\varphi)$: indeed,
\begin{equation}
  \gamma^\Phi_{\mu\nu} \varphi^\nu = \tfrac14 \varphi^\nu \varphi^\rho
  \epsilon_{\mu\nu\rho}{}^{\sigma\tau} r_{AB}
  \omega^{AB}_{\sigma\tau} = 0,
\end{equation}
by symmetry.  This means that we can choose $\lambda = 0$ and hence
$\alpha = 0$, $\beta = \beta^\Phi$ and $\gamma = \gamma^\Phi$ and the
Lie brackets of the Killing superalgebra are given by
\begin{equation}\label{eq:ksa-non-sugra}
  \begin{aligned}{}
    [A,B] &= AB - BA\\
    [A,s] &= \tfrac12 \omega_A \cdot s\\
    [A,v] &= Av
  \end{aligned}
  \qquad\qquad
  \begin{aligned}{}
    [s,s] &= \kappa(s) + \gamma(s,s)\\
    [v,s] &= \beta_v s\\
    [v,w] &= \rho(v,w),
  \end{aligned}
\end{equation}
subject to the Jacobi identity, which will determine $\rho$.

As in the case of the supergravity backgrounds, most of the components
of the Jacobi identity are already satisfied. The $[\fh,V,V]$ component
will follow from the $\fh$-component of the $[S,S,V]$ Jacobi. Indeed, this
component says that
\begin{equation}\label{eq:ssv-h}
  \rho(v,\kappa(s,s)) \stackrel{!}{=} 2 \gamma(\beta(v,s),s).
\end{equation}
If we can solve this equation for $\rho$, which basically means that the
RHS only depends on $\kappa(s,s)$ and not on either $\mu$ nor
$\omega^{AB}$, then $\rho$ is indeed $\fh$-equivariant, since so are
$\beta$ and $\gamma$, and $\rho$ maps to $\fh$, since so does $\gamma$.
The other two Jacobi components which need to be satisfied are the
$[S,V,V]$ component:
\begin{equation}
  \beta(v,\beta(w,s)) - \beta(w,\beta(v,s)) \stackrel{!}{=} [\rho(v,w),s]
\end{equation}
and the $[V,V,V]$ component:
\begin{equation}
  \rho(u,v) w +   \rho(v,w) u +   \rho(w,u) v  \stackrel{!}{=} 0.
\end{equation}
We actually prefer to derive $\rho$ from the $[S,V,V]$ Jacobi and check
the other two.  A straightforward (if tedious) calculation shows that
\begin{equation}
  \label{eq:rho}
  \rho_{\mu\nu\sigma\tau} =\varphi^2 g_{\mu[\sigma} g_{\tau]\nu} -
  \varphi_\mu \varphi_{[\sigma} g_{\tau]\nu} + \varphi_\nu
  \varphi_{[\sigma} g_{\tau]\mu},
\end{equation}
agreeing, as expected, with the Riemann
tensor~\eqref{eq:riem-non-sugra}.  It follows that $\rho$ is
$\fh$-equivariant and we check that it does map to $\fh$: indeed, a
short calculation shows that $\rho_{\mu\nu\sigma\tau} \varphi^\tau = 0$,
with terms cancelling pairwise. Similarly, one checks that
$\rho_{[\mu\nu\sigma]\tau} = 0$, which shows that the $[V,V,V]$ Jacobi
is satisfied and, finally, that equation~\eqref{eq:ssv-h} is too.

In summary, the Killing superalgebra of the maximally supersymmetry
backgrounds with $C=0$ is given by equation~\eqref{eq:ksa-non-sugra}
with $\fh = \fso(V) \cap \fstab(\varphi)$ and
\begin{equation}
  \begin{split}
    \beta_\mu{}^A{}_B &= -\tfrac18 \varphi^\nu \left( \Gamma_\mu\Gamma_\nu + 3 \Gamma_\nu \Gamma_\mu \right) r^A{}_B\\
    \gamma_{\mu\nu} &= \tfrac14 \epsilon_{\mu\nu\rho\sigma\tau} \varphi^\rho \omega^{\sigma\tau}\\
  \end{split}
\end{equation}
where we have introduced the shorthand $\omega_{\sigma\tau} := r_{AB}
\omega^{AB}_{\sigma\tau}$  and where $\rho$ is given by equation~\eqref{eq:rho}.

\section*{Acknowledgments}

We are very grateful to Paul de Medeiros and Andrea Santi for many
useful conversations about supersymmetry and Spencer cohomology.  We are
grateful to an anonymous referee for suggesting ways to improve the
exposition.  The research of AB is partially supported by an STFC
postgraduate studentship. Part of this work was previously submitted by
AB in a dissertation as part of an MSc degree under the University of
Edinburgh School of Physics and Astronomy and the Higgs Centre for
Theoretical Physics, with tuition fees paid by a Highly Skilled
Workforce Scholarship from the University of Edinburgh.

\appendix

\section{Tensorial identities for 2-forms}
\label{sec:tensor-id}

We collect here some identities which are required for various calculations in the main body of this paper. Let $(M,g)$ be a 5-dimensional lorentzian manifold and let $A,B\in\Omega^2(M)$. We work in a local orthonormal frame, starting with the identity
\begin{equation}
	\label{eq:tensor-id-AB-1}
	\epsilon_{\alpha\beta\gamma[\mu\nu}A\indices{_{\rho]}^\alpha}B^{\beta\gamma}
	+ \tfrac{2}{3}\epsilon_{\mu\nu\rho\alpha\beta}A^{\alpha\delta}B\indices{^\beta_\delta}
	=0,
\end{equation}
which can be verified by contracting the left hand side with $\epsilon^{\mu\nu\rho\sigma\tau}$. It then follows that
\begin{equation}\label{eq:tensor-id-AB-2}
	\epsilon_{\alpha\beta\gamma[\mu\nu}A\indices{_{\rho]}^\alpha}B^{\beta\gamma}
	 + \epsilon_{\alpha\beta\gamma[\mu\nu}A^{\alpha\beta}B\indices{_{\rho]}^\gamma} = 0.
\end{equation}
We can then show that
\begin{equation}
	\label{eq:tensor-id-AB-3}
	\epsilon_{\nu\rho\alpha\beta\gamma}\qty(
		A\indices{_{\mu}^\alpha}B^{\beta\gamma}+ A^{\alpha\beta}B\indices{_{\mu}^\gamma},
		)
	+ \eta_{\mu[\nu}\epsilon_{\rho]\alpha\beta\gamma\delta}A^{\alpha\beta}B^{\gamma\delta}
	=0
\end{equation}
by contracting the left hand side with $\epsilon^{\nu\rho\sigma\tau\chi}$ and using equation~\eqref{eq:tensor-id-AB-2}. In particular, we have
\begin{gather}\label{eq:tensor-id-AA-1}
	\epsilon\indices{_{\alpha\beta\gamma[\mu\nu}}
		A\indices{_{\rho]}^\alpha}A\indices{^{\beta\gamma}} = 0,
\\
	\label{eq:tensor-id-AA-2}
	\epsilon\indices{_{\alpha\beta\gamma\nu\rho}}
		A\indices{_{\mu}^\alpha}A\indices{^{\beta\gamma}}
	+ \tfrac{1}{2}\eta\indices{_{\mu[\nu}}\epsilon\indices{_{\rho]\alpha\beta\gamma\delta}}
		A\indices{^{\alpha\beta}}A\indices{^{\gamma\delta}} = 0.
\end{gather}


\providecommand{\href}[2]{#2}\begingroup\raggedright\endgroup

\end{document}